\documentclass[11pt]{article}
\usepackage{amsmath}
\usepackage{graphicx}
\usepackage{enumerate}
\usepackage{hyperref}

\usepackage{amsthm}
\usepackage{amssymb}
\usepackage[none]{hyphenat}
\usepackage{mathtools}
\usepackage{bm}

\usepackage[utf8]{inputenc}
\usepackage{subcaption}
\usepackage{float}
\usepackage[title]{appendix}

\usepackage{natbib}
\usepackage{authblk}

\usepackage{todonotes}
\setuptodonotes{color=blue!30}

\numberwithin{equation}{section}

\newcommand{\quotec}[1]{``#1"}
\newcommand{\E}[0]{\ensuremath{\mathtt{E}}-}

\newcommand{\hmc}[1]{$\mathcal{H}_#1$}

\newtheorem{proposition}{Proposition}
\newtheorem{theorem}{Theorem}

\theoremstyle{plain}
\newtheorem{definition}{Definition}

\DeclareMathOperator*{\argmin}{arg\,min}

\newcommand{\gro}{\ensuremath{\text{\sc gro}}}
\newcommand{\commentout}[1]{}

\newcommand{\trv}[2]{\ensuremath{#1^{(#2)}}}
\newcommand{\brv}[2]{\ensuremath{#1_{(#2)}}}
\newcommand{\cH}{\ensuremath{\mathcal{H}}}

\newcommand{\cW}{\ensuremath{\mathcal{W}}}
\newcommand{\cY}{\ensuremath{\mathcal{Y}}}
\newcommand{\naturals}{\ensuremath{\mathbf N}}
\newcommand{\reals}{\ensuremath{\mathbf R}}
\newcommand{\regret}{\ensuremath{\textsc{regret}}}

\newtheorem{examplehidden}{Example}
\newenvironment{example}{
\begin{examplehidden}
\em
}
{\end{examplehidden}}

\begin{document}
\bibliographystyle{apalike}

\def\spacingset#1{\renewcommand{\baselinestretch}%
{#1}\small\normalsize} \spacingset{1}


\title{\bf  Generic E-Variables for Exact Sequential $k$-Sample Tests that allow for Optional Stopping}
\author{Rosanne Turner\textsuperscript{a,b}\thanks{
\textsuperscript{a}CWI, Amsterdam, part of NWO-I \textsuperscript{b}University Medical Center Utrecht, Brain Center, \textsuperscript{c}University of Amsterdam, Department of Psychology and \textsuperscript{d}Leiden University, Department of Mathematics\\ Declarations of interest: none. *Corresponding author:
Rosanne Turner {\tt (Rosanne.Turner@cwi.nl)}
National Research Institute for mathematics and computer science in the Netherlands (CWI),
Science Park 123, 1098 XG Amsterdam,
The Netherlands
}, Alexander Ly\textsuperscript{a,c} and Peter Gr\"unwald\textsuperscript{a,d}}
\maketitle

\newpage
\paragraph{Abstract}
We develop \E variables for testing whether two or more data streams come from the same source or not, and more generally, whether the difference between the sources is larger than some minimal effect size. These \E variables lead to exact, nonasymptotic tests that remain safe, i.e. keep their type-I error guarantees, under flexible sampling scenarios such as optional stopping and continuation.
In special cases our \E variables also have an optimal `growth' property under the alternative. While the construction is generic, we illustrate it through the special case of $k \times 2$ contingency tables, where we also allow for the incorporation of different restrictions on a composite alternative. Comparison to p-value analysis in simulations and a real-world example show that \E variables, through their flexibility, often allow for early stopping of data collection  --- thereby retaining similar power as classical methods --- while also retaining the option of extending or combining data afterwards.

\newpage

\spacingset{1} 

\section{Introduction}
We develop hypothesis tests that are robust under flexible sampling scenarios, in which one is allowed to engage in optional continuation and optional stopping. We focus on the setting with data coming from several groups, the goal being to test whether the underlying distributions are all the same or not. Since it considerably simplifies notation and treatment, we focus on two-sample tests throughout the paper, pointing out at the relevant places how to extend our results to the $k$-sample setting for $k > 2$.
%
Our methods are based on  \E-variables and test martingales. While to some extent going back as far as   \cite{darling1967confidence}, interest in these concepts has exploded only very recently, in part in relation to the ongoing replicability crisis in the applied sciences 
 \citep{howard2018uniform,ramdas2020admissible,VovkW21,shafer2020language,grunwald2019safe,pace2019likelihood,manole2021sequential,henzi2021valid}.

Thus, suppose we collect samples from two distinct groups, denoted $a$ and $b$. In both groups, data are i.i.d. and come in sequentially --- even though, as explained underneath (\ref{eq:evalstart}) below, our approach can also be fruitfully used in the fixed design case. 
We thus have two data streams, $Y_{1,a}, Y_{2,a}, \ldots$ i.i.d. $\sim P_{\theta_a}$
and $Y_{1,b}, Y_{2,b},\ldots $ i.i.d. $\sim P_{\theta_b}$ with $\theta_a, \theta_b \in \Theta$, $\{P_{\theta} : \theta \in \Theta\}$ representing some parameterized underlying family of distributions, all assumed  to have a probability density or mass function denoted by  $p_{\theta}$ on some outcome space $\cY$. We will use notation $P_{(\theta_a,\theta_b)}$ (density $p_{(\theta_a,\theta_b)}$) to represent the joint distribution of both streams.  We consider a testing scenario, in which the null hypothesis
$\cH_0$
expresses that ${\theta_a}= {\theta_b}$ and the alternative $\cH_1$ expresses that $d(\theta_a,\theta_b) > \delta$ for some divergence measure $d$ and some effect size $\delta\geq 0$.  
We design a family of tests for this scenario that preserve type-I error guarantees under optional stopping. Hence, if the level $\alpha$-test is performed and the null hypothesis holds true, the probability that the null will {\em ever\/} be rejected is bounded by $\alpha$. Our tests can be implemented, and are exact, for arbitrary $\{ P_{\theta}: \theta \in \Theta \}$ and in combination with arbitrary divergence measures $d$. To our knowledge such a general construction is entirely new. For purposes of illustration and insight we choose to apply it to a very simple, classical problem: $2 \times 2$ contingency tables, with, in Section~\ref{sec:experiments}, an extension to $k\times 2$ tables. 
As is well-known (for completeness we provide simulations demonstrating this in the supporting information), if a standard fixed-design method for this scenario, the p-value resulting from Fisher's exact test, is (ab)used with optional stopping, the type-I error blows up. In contrast, our tests retain type-I error guarantee while, due to the optional stopping, having power competitive with Fisher's p-value. In fact, in the $k \times 2$ application (but not in general) our test has a GRO\footnote{Nonstandard abbreviations: GRO: growth-rate optimal; REGROW: relative growth-rate optimality in worst-case} {\em (growth-rate optimal)\/} property, GRO being the analogue of 'optimal power' in our optional continuation setting.

Our test depends on the choice of a prior distribution on the alternative $\cH_1 = \{P_{(\theta_a,\theta_b)}: (\theta_a,\theta_b) \in \Theta_1 \}$ with  $\Theta_1 \subset \{ (\theta_a,\theta_b) : \theta_a, \theta_b \in \Theta \}$. The choice of prior does not affect the type-I error safety guarantee, 
hence it is fine, even from a frequentist point of view, if such a prior is chosen based on vague prior knowledge. Still, the prior affects how fast one will tend to reject the null if it is indeed false. For the case that no clear prior knowledge is available, one may use the  prior 
that is optimal in terms of the relative GRO criterion; again the resulting test also has good power properties. 

\paragraph{\E-Variable Perspective; Block-wise Approach; Optional Continuation}
In its simplest form, an {\em \E variable\/} is a nonnegative random variable $S$ such that under all distributions $P$ in the null hypothesis,
\begin{equation}
    \label{eq:evalstart}
    {\bf E}_{P}[S] \leq 1. 
\end{equation}
Our test works by first designing \E variables for a  {\em single block\/} of data,  and then later extending these to sequences of blocks $Y_{(1)}, Y_{(2)}, \ldots$ by multiplication. A block is a set of data  consisting of $n_a$ outcomes in group $a$ and $n_b$ outcomes in group $b$, for some pre-specified $n_a$ and $n_b$. The $n_a$ and $n_b$ used for the $j$-th block $Y_{(j)}$ are allowed to depend on past data, but they must be fixed before the first observation in block $j$ occurs (this rule can be loosened to some extent, see Section~\ref{sec:simplesimple}). 

At each point in time, the running product of block \E variables observed so far is itself an \E variable, and the random process of the products is known as a {\em test martingale}. An \E variable-based test at level $\alpha$ is then a test with, in combination with any stopping rule $\tau$, reports `reject' if and only if the product of \E values corresponding to all blocks that were observed so far and have already been completed, is larger than $1/\alpha$. The full definition of $\tau$  may, and often will, be unknown to the user --- the user only needs to get the signal to stop and can then report the product \E variable. 
A classical paired one-sample test corresponds to the special case with $n_a = n_b = 1$ and data coming in in the order $a,b,a,b,\ldots$. 

We can combine \E variables from different trials that share a common null (but may be defined relative to a different alternative) by multiplication, and still retain type-I error control. If we used p-values rather than \E variables we would have to resort to e.g. Fisher's method for combining p-values, which, in contrast to multiplication of \( e \)-values, is invalid if there is a dependency between the (decision to perform) tests. With \E variables, such dependencies pose no problems for error control. Thus, in our setting, even if the design (i.e.  $n_a$ and $n_b$) is fixed in advance and optional stopping plays no role, we might still want to use the \E variable based tests described in this paper rather than a classic p-value based approach, since it allows us to do optional continuation over many experiments/studies (essentially, doing a meta-analysis \citep{grunwald2020safe}) while keeping type-I error control. 

\E variables and test martingales are explained in more detail in Section~\ref{sec:basics} below, but 
we refer to \cite{grunwald2019safe,shafer2020language} for an extensive introduction to \E variables, their use in `optional continuation' over several studies, and their enlightening {\em betting} interpretation.
The general story that emerges from these papers as well as, for example, \citep{VovkW21,ramdas2020admissible} is that \E variables and test martingales are the `right' generalization of likelihood ratios to the case that both $\cH_0$ and $\cH_1$ can be composite and combination of data from several trials may be required. 

\paragraph{Relevance}
Even in this age of big data and huge models, simple tests for comparing two populations are still used as heavily as ever in clinical trials, psychological studies and so on --- areas heavily plagued by the {\em reproducibility crisis\/} \citep{pace2019likelihood}. In a by-now notorious questionnaire \citep{john2012measuring}, more than $55\%$ of the interviewed psychologists admitted to the practice of `adding data until the results look good'. While classical methods lose their type-I error guarantee if one does this (Figure~\ref{fig:type1Error} in Appendix~\ref{app:experiments} in the Supporting Material),  \E-value based tests allow for it, while, due to the option of stopping early, remaining competitive in terms of sample sizes needed to obtain a desired power. We illustrate the practical advantage of our test  in Section~\ref{sec:realworld} using the  recent real-world example of the SWEPIS trial which was stopped early for harm \citep{wennerholm2019induction}. Their analysis being based on a p-value (by definition designed for fixed sampling plan), the question whether there was indeed sufficient evidence available to stop early is very hard to answer, since the sampling plan was not followed so that the p-value that led them to stop was by definition incorrectly calculated. This also makes it very difficult to combine the test results with results from earlier or future data while keeping anything like error control.
We show that with our \E-value based methodology we would have obtained sufficient evidence to stop for harm after the same number of events had occurred. Additionally, 
this \E-value, even though based on a stopped trial,  can 
be effortlessly combined with \E-values from other trials while retaining error guarantees. Also, our results are of interest beyond mere testing: the \E-variables we develop in this paper can be used to obtain {\em anytime-valid confidence intervals\/} \citep{howard2018uniform} that also remain valid under optional stopping. We will report on this extension elsewhere. 

SWEPIS summarized its data as a $2 \times 2$ contingency table. In Section~\ref{proportiontestS} and \ref{sec:compositeH1} we refine our generic test to the $2 \times 2$and $k \times 2$ model. An advantage of focusing on this simple setting  is that it is arguably the simplest and clearest example in which there is a nuisance parameter (the proportion under the null) that  does not admit a group invariance. Nuisance parameters that satisfy such an invariance (such as the variance in the $t$-test, or the grand mean in the two-sample $t$-test) are quite straightforward to turn into \E-variables and test martingales via the method of maximal invariants, as explained by \cite{grunwald2019safe} and
already put into practice by e.g. \cite{robbins1970statistical,lai1976confidence}. The present paper shows that the proportion under the null can also be handled in a clean and simple manner. 
As explained below, the resulting instantiated $2 \times 2$ test appears to be quite different from existing sequential and Bayesian approaches. Thus, more than 85 years after {\em the lady tasting tea}, we are able to still say something quite new about the age-old problem of contingency table testing.  

\paragraph{Related Work}
A sequential test for the $2 \times 2$ setting has been suggested as early as 1947 by Wald in his seminal \citep{Wald47} . Wald's test can be turned into a product of \E variables  and would then be safe to use under optional stopping. Yet, as  explained in Section~\ref{sec:conditional}, in the $2 \times 2$ setting the resulting \E variables do not grow as fast as the ones introduced here, and the underlying idea does not generalize to arbitrary models or effect size notions. Other earlier approaches (e.g. \cite[Section V.2]{siegmund2013sequential}) are based on asymptotic approximations. In contrast, our \E-variable based tests are  exact and nonasymptotic.
In fact our tests are more  closely related to, yet still different from, Bayes factor tests:
in the case of simple null hypotheses, \E-variable based tests coincide with Bayes factors \citep{grunwald2019safe}. However, in the $2 \times 2$ setting the null is not simple, and while the Bayes factor is a ratio of two Bayes marginal likelihoods, our \E-variables are ratios of more general, `prequential' \citep{Dawid84} likelihood ratios. In some special cases, the numerator is still  a Bayes marginal likelihood, but the denominator, in the $2 \times 2$ setting, almost never is (Section~\ref{sec:simplebayes}) . Thus, while similar in `look', our approach is in the end quite different from the default Bayes factors for tests of two proportions that were proposed by 
\cite{kass1992approximate} and by \cite{jamil2017default}, the latter based on early work by \cite{Dickey1974}.  To illustrate, in Appendix~\ref{app:gd} (Supplementary Material) we show that none of the variants of the Gunel-Dickey Bayes factor that are applicable in our set-up yield valid \E variables. 

Another, very recent, approach that bears some similarity to ours are the  two-sample tests from \cite{manole2021sequential}.
They focus on a nonparametric setting and their test martingales satisfy optimality properties as the sample size gets large.  Instead, we focus on the parametric case and, for this case, manage to derive  \E variables that are equal to or closely approximate the optimal (as measured according to the GRO criterion) \E variables, thus optimizing for the small-sample case (in principle, our tests could be used in a nonparametric setting as well, but since they rely on using a prior on the alternative, the test martingales of \cite{manole2021sequential} might be easier to use in that case). Another general nonparametric two-sample approach with a sequential flavor (but without optional stopping error guarantees) is \cite{lheritier2018sequential}.
\paragraph{Contents}
In the remainder of this introductory section, we formally introduce \E variables, optional stopping and the concept of GRO-optimality. In Section \ref{sec:twosample} we  propose our generic \E variable for tests of two streams in general and investigate when it has the GRO property. In Sections \ref{proportiontestS} and \ref{sec:compositeH1} we specifically show how these general \E variables can be applied in the setting of a test of two proportions, with and without  restrictions on the alternative hypothesis. In Sections \ref{sec:experiments} and \ref{sec:realworld} we provide, through simulations and a real-world example, comparisons of various \E variables and Fisher's exact test with respect to GRO and power.  In Section~ \ref{sec:otherevariables} we compare our generic approach to other \E variables one might define for this problem, including the ones based on Wald's section test. We end with a conclusion. All proofs are in the appendix.

\subsection{\E Variables and Test Martingales, Safety and Optimality}
\label{sec:basics}
We first need to extend the notion of \E variable to random processes: 
\begin{definition}\label{def:essential}
Let $\{\brv{Y}{j} \}_{ j \in \naturals }$, with all $\brv{Y}{j}$ taking values in some set $\cY$, represent a discrete-time random process.  
Let $\cH_0$ be a  collection of distributions for the  process $\{\brv{Y}{j} \}_{ j \in \naturals }$.   For all $j \in \naturals$, let  $\brv{S}{j}$ be a non-negative random variable that is adapted to $\sigma(\trv{Y}{j})$, with $\trv{Y}{j}= (Y_{(1)}, \ldots, Y_{(j)})$,   i.e. there exists a function $s$ such that $\brv{S}{j} = s(\trv{Y}{j})$. 
\begin{enumerate}
    \item 
We say that $\brv{S}{j}$  is an \E-variable for $\brv{Y}{j}$  conditionally on $\trv{Y}{j-1}$
if for all $P \in \cH_0$,
\begin{equation}\label{eq:eval}
  {\bf E}_P\left[ \brv{S}{j} \mid \brv{Y}{1}, \ldots, \brv{Y}{j-1} \right] \leq 1.
\end{equation}
That is, for each $\trv{y}{j-1}{} \in \mathcal{Y}^{j-1}$, all $P_0 \in \cH_0$, (\ref{eq:evalstart}) holds with $S=s(\brv{y}{1},\ldots,\brv{y}{j-1},\brv{Y}{j})$ and $P$
set to $P_0 \mid \trv{Y}{j-1}= \trv{y}{j-1}$. 
\item If, for each $j$, $\brv{S}{j}$ is an \E variable conditional on $\brv{Y}{1}, \ldots, \brv{Y}{j-1}$, then we call the process $\{\brv{S}{j}{}\}_{j \in \naturals}$ a {\em sequential \E variable process\/} relative to the given $\cH_0$ and $\{\brv{Y}{j} \}_{j \in \naturals}$ and we call 
$\{\trv{S}{m}{}\}_{m \in \naturals}$ with $\trv{S}{m}{} = \prod_{j=1}^m \brv{S}{j}$ the corresponding {\em test martingale}.   
\end{enumerate}
\end{definition}
Henceforth, we omit the phrase `relative to $\cH_0$ and $\{ \brv{Y}{j} \}_{j \in \naturals}$' whenever it is clear from the context. 
By the tower property of conditional expectation, one verifies that for any process of conditional \E variables $\{\brv{S}{j} \}_{j \in \naturals }$, we have for all $m$ that the product $\trv{S}{m}{}$ is itself an `unconditional' \E variable as in (\ref{eq:evalstart}), i.e. ${\bf E}_{P}[S^{(m)}] \leq 1$ for all $P \in \cH_0$. 
Definition~\ref{def:essential} adapts and slightly modifies terminology from \citep{shafer2011test}. As follows from that paper, in standard martingale terminology, what we call a test martingale is a non-negative supermartingale relative to the filtration induced by $\{\brv{Y}{j} \}_{j \in \naturals }$, with starting value $1$.

\paragraph{Safety}
The interest in \E variables and test martingales derives from the fact that we have type-I error control irrespective of the stopping rule used: for any test martingale $\{\trv{S}{j}\}_{j \in \naturals}$, Ville's inequality \citep{shafer2020language} tells us that, for all $0 < \alpha \leq 1$, $P \in \cH_0$,
\begin{equation}\label{eq:ville}
P(\text{there exists $j$ such that\ } \trv{S}{j} \geq 1/\alpha) \leq \alpha.
\end{equation}
Thus, if we measure evidence against the null hypothesis after observing $j$ data units by $\trv{S}{j}$, and we reject the null hypothesis if $\trv{S}{j} \geq 1/\alpha$, then our type-I error will be bounded by $\alpha$, no matter what stopping rule we used for determining $j$.
We thus have type-I error control even if we use the most aggressive stopping rule compatible with this scenario, where we stop at the first $j$ at which $S^{(j)} \geq 1/\alpha$ (or we run out of data, or money to generate new data). We also have type-I error control if the actual stopping rule is unknown to us, or determined by external factors independent of the data $\brv{Y}{j}$.

We will call any test based on $\{\trv{S}{j}\}_{j \in \naturals}$ and a (potentially unknown) stopping time $\tau$ that, after stopping, rejects iff $\trv{S}{\tau} \geq 1/\alpha$ a {\em level $\alpha$-test that is safe under optional stopping}, or simply a {\em safe test}.

\begin{example}\label{ex:first}
Let $P_0$ and $Q$ be any two distributions for the process $Y_{(1)}, Y_{(2)}, \ldots$, and let $\cH_0 = \{ P_0 \}$ represent a simple null.
Let $S^{(m)}$ denote the likelihood ratio for $m$ outcomes and $S_{(j)}$ its constituent factors, i.e. 
\begin{equation}
S^{(m)} = \frac{
q(Y^{(m)})}{p_0(Y^{(m)})} = \prod_{j=1}^m S_{(j)} \text{\ with\ }
S_{(j)} = \frac{q(Y_{(j)}\mid Y^{(j-1)})}{
p_0(Y_{(j)}\mid Y^{(j-1)})}
\end{equation}
where $q(y_{(m)} \mid y^{(m-1)})$ denotes the conditional density corresponding to $Q$ and $p_0(y_{(m)} \mid y^{(m-1)})$ the one corresponding to $P_0$ with respect to a common underlying measure. 
Then the likelihood ratio process $\{S^{(m)} \}_{m \in \naturals}$ constitutes a test martingale, and the process of past-conditional likelihoods $\{S_{(j)} \}$ is a sequential \E variable process relative to $\cH_0$. This can be immediately verified by directly calculating the conditional expectation of $S_{(j)}$ given $Y^{(j-1)}$, noticing that the densities $p_0(Y_{(j)} |Y^{(j-1)})$ cancel in the calculation.
\end{example}
\paragraph{GRO-Optimality, Simple $\cH_1$}
Just like for p-values, the definition of \E variables only requires
explicit specification of $\cH_0$, not of an alternative hypothesis $\cH_1$.  $\cH_1$ becomes crucial  once we distinguish between `good' and `bad' \E variables:
\E variables have been designed to remain small, with high probability, under
the null $\cH_0$. But if $\cH_1$ rather than $\cH_0$ is true, then `good'
\E variables should produce evidence (grow --- because the larger the
\E variable, the closer we are to rejecting the null) against $\cH_0$ as fast
as possible. To make this precise, first consider simple (singleton) $\cH_1= \{ Q \}$. We start with the one-outcome setting of (\ref{eq:evalstart}), i.e. we look at a single \E variable $S_{(j)}$ in isolation for a single outcome $Y_{(j)}$.  Its optimality 
is  measured in terms of
\begin{equation}\label{eq:grow}
{\bf E}_{Q}[\log S_{(j)}],\end{equation}
and the
\E variable which maximizes this quantity among all \E variables that can be written as functions of $Y_{(j)}$ (i.e. non-negative random variables satisfying (\ref{eq:evalstart})), assuming it exists, 
is called the
\emph{Growth Rate Optimal} \E variable for $Y_{(j)}$ relative to $Q$, or simply `$Q$-GRO for $Y_{(j)}$', and denoted as $S_{\gro(Q),(j)}$
More generally, \E variable  $\trv{S}{m}$ is called  \emph{growth rate optimal} relative to $Q$ for $Y^{(m)}$, or simply $Q$-GRO for $Y^{(m)}$, if, among all (unconditional) \E-variables that can be written as a function of $\trv{Y}{m}$, it maximizes  
\begin{equation}\label{eq:marggrow}
{\bf E}_{Q}[\log \trv{S}{m}].
\end{equation} 
We will denote this 
\E-variable, if it exists, by $S^{(m)}_{\gro(Q)}$. The idea to maximize (\ref{eq:marggrow}) goes back to \cite{Kelly56}; the GRO-terminology is from \cite{grunwald2019safe}.
The larger an \E-variable or test martingale tends to be under the alternative, the better it scores in the GRO sense. Of course, the same would still hold if we were to replace the logarithm by another strictly increasing function. But there are various compelling reasons for
why one should take a logarithm here --- see  \cite{grunwald2019safe,shafer2020language}. One interesting reason, not explicitly covered by these two papers, was already given by \cite{breiman1961optimal} and is explained  in detail by \cite[Appendix B.1]{grunwald2020safe}: the $Q$-GRO test martingale
is
also the test martingale which minimizes the expected number of data points needed
before the null can be rejected if we use the test with the aggressive stopping rule described before (reject at the smallest $j$ such that $S^{(j)} \geq 1/\alpha)$. Thus, using the $Q$-GRO test martingale is quite analogous to employing a test that maximizes power. One can also directly see that both notions must be connected by noting that GRO implies optimizing the expectation of $\log S^{(j)}$ whereas power at fixed sample size $j$ is the probability that $\log S^{(j)}$  is larger than $- \log \alpha$. 
Note that we cannot directly use power in designing tests, since the notion of power requires a fixed sampling plan, which we will usually not have: we may not want or not be able to stop at the first $j$ such that we can reject --- for example, we might want to stop early for harm (Section~\ref{sec:realworld}), or we might want to lower $\alpha$ if the first few outcomes look very promising. So we will measure optimality in terms of GRO instead, but for practical usefulness we  {\em do\/} hope that, in cases where we do follow the sampling plan above (stop as soon as  $S^{(j)} \geq 1/\alpha)$), our power remains reasonable. This is suggested by Breiman's observation above, but we want to check it nevertheless. Such a check is done successfully for the $2 \times 2$ model in Section~\ref{sec:experiments}.

In `nice' cases, the $Q$-GRO \E-variable (\ref{eq:marggrow}) for $m$ outcomes can be obtained by multiplying the individual $Q$-GRO \E-variables: 
\begin{proposition}\label{prop:proper}
Let $\cH_1 = \{Q\}$ be simple and $\cH_0$ be potentially composite, and `nondegenerate' in the sense that for some $P \in \cH_0$, $D(Q\| P) < \infty$, $D(\cdot \| \cdot)$ denoting the KL divergence. Suppose the following condition holds (with $q$, $p$ the density of $Q$ and $P$, respectively):
\begin{equation}\label{eq:thecondition}
   \text{
   There exists a $P\in \cH_0$ such that $S_{(1)} = q(Y_{(1)})/p(Y_{(1)})$ is an \E-variable. } 
\end{equation}
Then
$S_{(1)} = S_{\gro(Q),(1)}$ is the  $Q$-GRO \E variable for $Y_{(1)}$. 
An \E variable of this form automatically exists if $\cH_0$ is simple.
%
If we further assume that $Y_{(1)}, Y_{(2)}, \ldots$ are i.i.d. according to all distributions in $\cH_0 \cup \cH_1$, then $S^{(m)}_{\gro(Q)} = \prod_{j=1}^m S_{\gro(Q),(j)}$, i.e. the $Q$-GRO optimal (unconditional) \E variable for $Y^{(m)}$ is the product of the individual $Q$-GRO optimal \E variables. 
\end{proposition}
If Condition~(\ref{eq:thecondition}) holds and $Y_{(1)}, Y_{(2)}, \ldots$ are i.i.d. according to all distributions in $\cH_0 \cup \cH_1$, it thus makes sense to define the {\em $Q$-GRO  test martingale\/} to be the test martingale $(S^{(j)}_{\gro(Q)})_{j \in \naturals}$. 
We will then have that $S_{\gro(Q),(j)} = s_Q(Y_{(j)})$ for a fixed function $s_Q: \cY \rightarrow \reals^+_0$. 
\begin{example}{\bf [Simple $\cH_1$ and Simple $\cH_0$]}\label{ex:second}
Consider $\cH_1 = \{Q \}$ and simple $\cH_0 = \{P_0 \}$ and arbitrary $Q'$ such that the $Y_{(j)}$ are i.i.d. according to $P, Q$ and $Q'$. Then $S_{(j)} = q'(Y_{(j)})/p_0(Y_{(j)})$ is an \E variable for $Y_{(j)}$, irrespective of the definition of $Q'$, by the same argument as in Example~\ref{ex:first}. By the Proposition above, the $Q$-GRO \E variable for $Y_{(j)}$ is given by setting $q'= q$. Then ${\bf E}_Q
[S_{\gro(Q),(j)}] = {\bf E}_{Y_{(j)} \sim Q}[\log q(Y_{(j)})/p_0(Y_{(j)})]$ also coincides with the KL divergence between $Q$ and $P_0$. 
\end{example}
In Section~\ref{sec:twosample} (Theorem~\ref{simpleSproportions}) we develop functions $s_{Q}$ (denoted $s(\cdot ; n_a, n_b,\theta^*_a,\theta^*_b)$ there) for simple $\cH_1= \{ Q \}$ so that $S_{Q,(1)} = s_{Q}(Y_{(1)})$
is an \E-variable even though $\cH_0$ is composite and not convex, so that Proposition~\ref{prop:proper} does not apply. 
Since we invariably assume the $Y_{(j)}$ are i.i.d., $S_{Q,(j)} := s_{Q}(Y_{(j)})$ is an \E-variable as well and with $S^{(m)}_{Q} := \prod_{j=1}^m S_{Q,(j)}$, $(S^{(m)}_{Q})_{m \in \naturals}$ is a test martingale. 
The construction works for the general setting of two data streams discussed in the introduction, and for some special $\cH_0$ (even though composite and nonconvex), the $S_{Q,(j)}$ will in fact be $Q$-GRO and $(S^{(m)}_{Q})_{m \in \naturals}$ will be the $Q$-GRO test martingale. These include the $\cH_0$ that arise in the $2 \times 2$ setting, our main application. 
For other $\cH_0$, the $\E$ variables $S_{Q,(j)}$ will not necessarily have the $Q$-GRO-property; they are  designed to have (\ref{eq:marggrow}) large, but it may be even larger for other \E variables.

\paragraph{GRO and Composite $\cH_1$}
In case $\cH_1$ is composite, no direct analogue of the GRO-criterion for designing \E variables exists, since it is not clear under what distribution $Q \in \cH_1$ we should maximize (\ref{eq:marggrow}).
In this paper, we deal with this situation by {\em learning\/} $Q$ from the data in a Bayesian fashion. It is now convenient to write $\cH_1 = \{P_{\theta} : \theta \in \Theta_1 \}$ in a parameterized manner (accordingly, henceforth we shall write  $\theta_1$-GRO \E variable instead of $P_{\theta_1}$-GRO \E-variable and $S_{\gro(\theta),(j)}$ instead of $S_{\gro(P_{\theta}),(j)}$). We will assume i.i.d. data, thus, if $\cH_1$ were true, then data would be i.i.d. $\sim P_{\theta^*_1}$ for some $\theta^*_1 \in \Theta_1$. 
Starting with a distribution $W$ on $\Theta_1$, i.e. a prior, at each point in time $j$, we determine the Bayesian posterior $W \mid Y^{(j-1)}$ and use the Bayes predictive $P_{W \mid Y^{(j-1)}}:= \int_{\Theta_1} P_{\theta} d W(\theta \mid Y^{(j-1)})$ as an estimate for the `true' $P_{\theta^*_1}$.
 As is well-known, under conditions on $W$ and $\cH_1$ (which, if $\cH_1$ is finite-dimensional parametric, are very mild), the posterior will concentrate around ${\theta^*}$ and hence $P_{W \mid Y^{(j-1)}}$ will resemble $P_{\theta^*_1}$ more and more, with very high probability, as more data becomes available. 

At each point in time $j$, we use our current estimate $P_{W \mid Y^{(j-1)}}$ to 
design a conditional \E variable $S_{(j)}$. On an informal level, as long as $P_{W \mid Y^{(j-1)}}$ converges to the `true' $P_{\theta^*_1}$, 
the  $S_{(j)}$ will in fact also start to more and more resemble the \E-variables $S_{\gro(\theta^*_1),(j)}$ we designed for  $\cH_1= \{ P_{\theta^*_1} \}$ and which were designed to have a large expected growth under the `true' $P_{\theta^*_1}$. Assuming the convergence happens fast, we have that
\begin{align}
    \label{eq:jasaregret}
{\bf E}_{Y^{(m)} \sim P_{\theta^*_1}} \left[ \log S^{(m)}_{\gro(\theta^*_1)}     - \log \prod_{j=1}^m S_{(j)}\right]
 \end{align}
is small, i.e. we may  expect that the test martingale $\prod_{j=1}^m S_{(j)}$ grows not much slower than $S^{(m)}_{\gro(\theta^*_1)} = \prod_{j=1}^m S_{\gro(\theta^*_1),(j)}$, the best test martingale (maximizing ${\bf E}_{Y^{(m)} \sim P_{\theta^*_1}} \left[ \log S \right]$ over all \E variables $S$ for $Y^{(m)}$) we could have used if we had known the true $P_{\theta^*_1}$ all along. 


\section{Two-Stream Safe Tests}\label{sec:twosample}
Consider the two-stream setting introduced in the beginning of the paper. To formalize it further, we introduce {\em calendar time\/} $t=1,2, \ldots$ and corresponding random variables $V_t$ and $G_t$: at each $t$, we obtain an outcome $V_t$ in $\cY$ in group $G_t \in \{a,b\}$. Importantly though, at this point we make no assumptions about the relative ordering of outcomes from the two groups. At time $t$, we have that $t_a$, the number of $a$'s that are observed so far, and $t_b$, the number of $b$'s observed so far, satisfy $t_a+t_b = t$, but subject to this constraint we allow them coming in any order, e.g. first all $a$'s, or first all $b$'s, or interleaved.
For example, with $t_a = 3$ and $t_b =2$, we might have $V_1 = Y_{1,a}, V_2 = Y_{2,a}, V_3 = Y_{3,a}, V_{4} = Y_{1,b}, V_{5} = Y_{2,b}$ (all $a$s come first, $G_1 = G_2 = G_3 = a, G_4=G_5 =b$) but also, for example $V_1 = Y_{1,a}, V_2 = Y_{1,b}, V_3 = Y_{2,a}, V_{4} = Y_{3,a}, V_{5} = Y_{2,b}$.

We thus have that the (marginal) probability of the first $t= t_a+ t_b$ outcomes, given that $t_a$ of these are in group $a$ and $t_b$ in group $b$, and writing $y^t = (y_1, \ldots, y_t)$, is given by the probability density (or mass function)
\begin{equation} \label{eq:firstgeneral}
p_{\theta_a,\theta_b} (y_{a}^{t_a}, y^{t_{b}}_b) \coloneqq 
p_{\theta_a} (
y_{a}^{t_a}) p_{\theta_b}(
y^{t_{b}}_b) = \prod_{t=1}^{t_a} p_{\theta_a}(y_{t,a})
\prod_{t=1}^{t_b} p_{\theta_b}(y_{t,b}). 
\end{equation}
To indicate that random vector $(Y^{t_a}_a,Y^{t_b}_b) \coloneqq (Y_{1,a} \ldots, Y_{t_a,a},Y_{1,b}, \ldots, Y_{t_b,b})$ has a distribution represented by (\ref{eq:firstgeneral}) we write `$Y^{t_a}_a, Y^{t_b}_b \sim P_{\theta^*_a,\theta^*_b}$'. 

According to the {\em null hypothesis\/} $\cH_0 = \{P_{\theta_a,\theta_b} : (\theta_a,\theta_b) \in \Theta_0 \}$, $\Theta_0 = \{ (\theta,\theta): \theta \in \Theta \}$, both processes coincide. Thus, we have that $\theta^*_a = \theta^*_b = \theta_0$ for some $\theta_0 \in  \Theta$ and then the density of data $y_{a}^{t_a}, y^{t_{b}}_b$ is given by $p_{\theta_0}(y_{1,a}, \ldots, y_{t_a,a},y_{1,b}, \ldots, y_{t_b,b})$.

\subsection{A generic \E variable for 2-stream--blocks}
\label{sec:simplesimple}
We first consider the case in which the alternative hypothesis is simple: $\Theta_1 = \{\theta_1\}$ for some fixed $\theta_1 = (\theta_a^*, \theta_b^*) \in \Theta^2$. 
Consider a fixed sample size of size $n$, and assume that we will observe a block of $n_a$ outcomes in group $a$ and $n_b$ outcomes in group $b$.
In this case, we can define an \E variable as the likelihood ratio between $p_{\theta^*_a,\theta^*_b}$ and a carefully chosen distribution that is a product of mixtures of distributions from $\Theta_0$: for 
$n_a, n_b \in \naturals$, $n \coloneqq n_a + n_b$ and $y^{n_a}_a = (y_{1,a}, \ldots, y_{n_a,a}) \in \cY^{n_a}$ and $y^{n_b}_b = (y_{1,b}, \ldots, y_{n_b,b}) \in \cY^{n_b}$, we define:
\begin{multline}\label{eq:simpleevar}
s(y^{n_a}_a, y^{n_b}_b; n_a, n_b, \theta^*_a,\theta^*_b) 
\coloneqq \\ \frac{p_{\theta_{a}^*}( y^{n_a}_a) }{
\prod_{i=1}^{n_a} 
 \left( \frac{n_a}{n} p_{\theta^*_a} (y_{i,a}) + \frac{n_b}{n} p_{\theta^*_b} (y_{i,a}) \right)} \cdot 
 \frac{p_{\theta_b^*}(y^{n_b}_b)}{
\prod_{i=1}^{n_b} 
 \left( \frac{n_a}{n} p_{\theta^*_a} (y_{i,b}) + \frac{n_b}{n} p_{\theta^*_b} (y_{i,b}) \right)} .
 \end{multline} 
\begin{theorem}
\label{simpleSproportions}
The  random variable $S_{[n_a,n_b, \theta^*_a,\theta^*_b]} :=  s(Y^{n_a}_a, Y^{n_b}_b; n_a,n_b, \theta^*_a,\theta^*_b)$ is an \E variable, i.e. 
we have: 
\begin{equation*}
\sup_{\theta \in \Theta} {\bf E}_{V^n \sim  P_{\theta}}\left[
s( V^n ; n_a,n_b, \theta^*_a,\theta^*_b)
\right] \leq 1.
\end{equation*}
Moreover, if $\{ P_{\theta}: \theta \in \Theta \}$  is a  convex set of distributions, then $S_{[n_a,n_b, \theta^*_a,\theta^*_b]}$ is the $(\theta^*_a,\theta^*_b)$-GRO \E variable: for any non-negative function $s'$ on $\cY^{n_a+n_b}$ satisfying
$\sup_{\theta \in \Theta} {\bf E}_{V^n \sim  P_{\theta}}\left[
s'( V^n)\right] \leq 1$, we have:
\begin{equation*}
{\bf E}_{Y^{n_a}_a, Y^{n_b}_b  \sim P_{\theta^*_a,\theta^*_b}}[ \log 
s(Y^{n_a}_a, Y^{n_b}_b; n_a,n_b, \theta^*_a,\theta^*_b) ] \geq 
{\bf E}_{Y^{n_a}_a, Y^{n_b}_b  \sim P_{\theta^*_a,\theta^*_b}}[ \log s'(Y^{n_a}_a, Y^{n_b}_b)].
\end{equation*}
\end{theorem}
Crucially, in the second part of the theorem, we do not require convexity of $\cH_0$, a set of distributions over $\cY^{n_a+n_b}$ (if $\cH_0$ were convex, the GRO property would already follow automatically \citep{KoolenG21}), but instead of $\{P_{\theta}: \theta \in \Theta \}$, a set of distributions on $\cY$. In the $2 \times 2$ case $\cH_0$ is not convex, since the set of i.i.d. Bernoulli distributions over $n_a+n_b > 1$ outcomes is not convex; but $\{P_{\theta}: \theta \in \Theta \}$ is just the Bernoulli model on one outcome, which is convex, so that in this setting, we get the GRO \E variable. 

To illustrate, consider the basic case in which data comes in in fixed batches $Y_{(1)}, Y_{(2)}, \ldots$, with each batch $Y_{(j)} = ((Y_{(j-1)n_a +1,a},Y_{(j-1) n_a +2,a},\ldots, Y_{j n_a,a}) $ $
,(Y_{(j-1)n_b +1,b},Y_{(j-1) n_b +2,b},\ldots, Y_{j n_b,b}) ), $ having exactly $n_a$ outcomes in group $a$ and $n_b$ outcomes in group $b$, and let $n = n_a+ n_b$. This case would obtain, for example, in a sequential clinical trial in which patients come in one by one, each odd patient is given the treatment and each even patient is given the placebo. Then $n=2$, $n_a=n_b =1$. We may then measure the evidence against the null hypothesis by the product E-value
\begin{equation}\label{eq:fourth}
S^{(m)}_{[n_a,n_b,\theta^*_a,\theta^*_b]} \coloneqq 
\prod_{j=1}^m S_{(j),[n_a,n_b,\theta^*_a,\theta^*_b]} \ \ ; \ \ 
S_{(j),[n_a,n_b,\theta^*_a,\theta^*_b]}
\coloneqq s(Y_{(j)}; n_a, n_b, \theta^*_a,\theta^*_b).
\end{equation}
By Ville's inequality (\ref{eq:ville}), the probability under any distribution in the null that there is an $m$ 
with $S^{(m)}_{[n_a,n_b, \theta^*_a,\theta^*_b]}$ larger than $1/\alpha$, is bounded by $\alpha$, hence, type-I error guarantees are preserved under optional stopping if we perform the test based on $\{ S^{(m)}_{[n_a,n_b,\theta^*_a,\theta^*_b]}\}_{m \in \naturals}$ as defined underneath (\ref{eq:ville}), as long as we stop between and not `within' batches (if we stop within a batch, the E-variable $S^{(m)}_{[n_a,n_b, \theta^*_a,\theta^*_b]}$ is undefined). 

If the data do not come in batches of equal size, we may proceed as follows. First, we need to fix some $n_a\geq 1 $ and $n_b \geq 1$ of our own choice. The treatment below will give valid \E variables irrespective of our choice of $n_a$ and $n_b$, but it will be seen that some choices are much more reasonable (will lead to much more evidence against the null, if the null is false) than others. 

Thus, fix $n_a$ and $n_b$, set $n = n_a + n_b$. At each time $t$, we will have observed, so far, some number $t_a$  of outcomes in group $a$, and $t_b$ in group $b$. Now let $m_t$ be the largest $m$ such that $m n_a \leq t_a$ and $m n_b \leq t_b$. Now, for $m=1,2, \ldots$, define $Y_{(m)}$ as above. At any given time $t$, $Y_{(1)}, Y_{(2)}, \ldots, Y_{(m_t)}$ will have been observed, and there may be a number $n'_j$ remaining observations in group $j \in \{a,b\}$ so that either $n'_a < n_a$ or $n'_b <n_b$ or both. Since the $\{\brv{Y}{j} \}_{j \in \naturals }$ determine a test martingale in the sense of Definition~\ref{def:essential}, optional stopping while preserving type-I error guarantees is then possible at any point in time $t$, as long as the \E variable is calculated as (\ref{eq:fourth}) above for $m = m_t$, thus ignoring the final $n'_a + n'_b$ outcomes.

How should $n_a$ and $n_b$ be chosen in practice? For example, consider a variation of the clinical trial setting above in which the treatment-control assignment is randomized: for each incoming patient, a fair coin is flipped to decide treatment $(a)$ or placebo $(b)$. 
Then at any given time the number of patients in group $a$ and $b$ will not be precisely equal, but if we choose $n_a = n_b = 1$ as above it is highly unlikely that the amount of data we have to ignore at any given time $t$ is very large. Similarly, if $G_t$, the group membership of the $t$-th observation is itself i.i.d. according to some distribution $P^*$, we might have some idea of the probability $p^*(a)$ assigned to group $a$; if $p^*(a) = 2/5$ (say), we would choose
$n_a =2, n_b =3$. 

We can add a significant amount of extra flexibility by allowing for variable group sizes, i.e., the chosen $n_a$ and $n_b$ may depend on the past. For this, we introduce a function $f: \bigcup_{t \geq 0} \cY^t \times \{0,1\}^t \rightarrow \{\text{\sc stop-block},\text{\sc continue}\}$ that, at each point in time $t$, decides whether the current block should end $(f(V^t, G^t) = \text{\sc stop-block})$ or not $(f(V^t,G^t) = \text{\sc continue})$. 
As long as the value of this function does not depend on the actual outcomes $V_t$ observed after the last block that was completed, all requirements for having a test martingale and thus for safe optional stopping are met. For example, suppose that on data $V_1, G_1, V_2, G_2, \ldots, V_t, G_t$ observed so-far, $f$ has output $\text{\sc stop-block}$ at $m$ occasions, the last time at $t'= t-k$ for some $k > 0$. Then $f(t)$ is allowed to depend on $Y^{(m)}$ and $G^t$, but for any fixed $Y^{(m)}= y^{(m)}, G^t = g^t$, for all $y^k, y'^k \in \cY^{k}$, we must have $f((y^{(m)}, y^k),g^t) = f((y^{(m)},y'^k),g^t)$. 
In this way, one can in principle  {\em learn\/} $p^*(a)$ from the data, changing group sizes $n_a$ and $n_b$ flexibly as data come in. For simplicity, we have not followed this approach here, but all our results readily extend to this case. 
\paragraph{Extension to $k$-sample streams} It is entirely straightforward
to extend \eqref{eq:simpleevar} to the scenario where we do not compare $2$, but $k$ i.i.d. data streams. 
Indeed, in the appendix we state and prove the generalization of  Theorem~\ref{simpleSproportions} to $k$ data streams. 
We again consider some fixed $\vec{\theta} = (\theta_a, \theta_b, ..., \theta_k) \in \Theta^k$. The probability of the first $t = \sum_{g = 1}^k t_g$ outcomes is now given by the density or mass function $p_{\vec{\theta}} \coloneqq p_{\theta_a}(y_a^{t_a})p_{\theta_a}(y_b^{t_b}) ... p_{\theta_k}(y_k^{t_k})$. We now need to fix the $k$ group outcome numbers $\vec{n} \coloneqq (n_a, n_b, ..., n_k)$ in advance, which allows us to define the extended \E variable as a function of the data $\vec{y}^n = (y_a^{n_a}, y_b^{n_b}, ..., y_k^{n_k})$, with $n= \sum_{g=1}^k n_g$: 
\begin{equation}\label{eq:evarkextension}
s(\vec{y}^n; \vec{n}, \vec{\theta}^*) 
\coloneqq \prod_{g = 1}^k
\frac{p_{\theta_{g}^*}( y^{n_g}_g) }{
\prod_{i=1}^{n_g} 
 \left( \sum_{g' = 1}^k \frac{n_{g'}}{n} p_{\theta^*_{g'}} (y_{i,g}) \right)} ,
 \end{equation}
 for testing the null where $\theta_a = \theta_b = ... = \theta_k$; it is again GRO if $\{P_{\theta}: \theta \in \Theta\}$ is convex. 
 We now return to the notationally simpler 2-sample case except for a short example of an application of this extension as a flexible and exact alternative to the chi-square test in section \ref{sec:experiments}.

\subsection{The generic \E variable with Bayesian alternative}
\label{sec:simplebayes}
Now fix some prior $W_1$ with density $w_1$ on the alternative $\Theta_1 \subseteq \Theta^2$. 
We can trivially extend the definition of our generic \E-variable relative to singleton $(\theta^*_a,\theta^*_b)$ to an \E-variable relative to arbitrary prior $W_1$ on $(\theta^*_a,\theta^*_b)$: define $p_{W_1,a}(y) := \int p_{\theta_a}(y) d W_1(\theta_a)$, the integration being over the marginal prior distribution over $\theta_a$, and similarly, $p_{W_1,b}(y) := \int p_{\theta_b}(y) d W_1(\theta_b)$.
Then, as a corollary of Theorem~\ref{simpleSproportions}, 
\begin{multline}\label{eq:simplewithprior}
s(y^{n_a}_a, y^{n_b}_b; n_a, n_b, W_1) \coloneqq \\
\frac{\prod_{i=1}^{n_a}  p_{W_{1,a}}( y_{i,a}) }{
\prod_{i=1}^{n_a} 
 \left( \frac{n_a}{n} p_{W_{1,a}} (y_{i,a}) + \frac{n_b}{n} p_{W_{1,b}} (y_{i,a}) \right)} \cdot 
 \frac{\prod_{i=1}^{n_b} p_{W_{1,b}}(y_{i,b})}{
\prod_{i=1}^{n_b} 
 \left( \frac{n_a}{n} p_{W_{1,a}} (y_{i,b}) + \frac{n_b}{n} p_{W_{1,b}} (y_{i,b}) \right)} .
 \end{multline}
is itself also an \E-variable, as follows from applying Theorem~\ref{simpleSproportions} with a `meta'-set of distributions,
which is possible since we made no assumptions at all on the set $\Theta$ in Theorem~\ref{simpleSproportions}: we replace $\Theta$ by $\cW(\Theta)$, the set of distributions on $\Theta$;  we replace the  background set of distributions $\{ p_{\theta}: \theta \in \Theta\}$  by the set of distributions $\{ p_{W}: W \in \cW(\Theta) \}$; we replace the simple $\cH_1 = \{P_{\theta^*_a,\theta^*_b} \}$ by a `simple' $\cH'_1 = \{ P_{W_a,W_b}\}$ for some distributions $W_a$ and $W_b$ on $\Theta$.
Such $W_1$-based generic \E-variables can be used to {\em learn\/} the parameters $\theta^*_a,\theta^*_b$ as more data in both streams come in, and this is how we will use them in a sequential context with optional stopping. Thus, assume again that data comes in batches $Y_{(1)}, Y_{(2)}, \ldots$ with each $Y_{(j)}$ consisting of $n_a$ outcomes in group $a$ and $n_b$ outcomes in group $b$ (generalization to flexible group sizes changing in time and depending on the past as described at the end of Section~\ref{sec:simplesimple} is straightforward). 
We start with some prior $W_1$ for the first batch $Y_{(1)}$ but we now use, for the $j$-th batch $Y_{(j)}$, the {\em Bayesian posterior\/} $W_1 \mid Y^{(j-1)}$ as prior to define the  $j$-th \E-variable with:
\begin{equation}\label{eq:sixth}
S^{(m)}_{[n_a,n_b,W_1]} \coloneqq \prod_{j=1}^m 
S_{(j),[n_a,n_b,W_1]} \ \ ; \ \ 
S_{(j),[n_a,n_b,W_1]} \coloneqq
s(Y_{(j)}; n_a, n_b, W_1 | Y^{(j-1)}).
\end{equation}
Again, $\{ S_{(j),[n_a,n_b,W_1]} \}_{j \in \naturals}$ is a sequential \E-variable process, so testing based on the corresponding test martingale is safe under optional stopping by (\ref{eq:ville}). If data are sampled from some alternative hypothesis $(\theta^*_a,\theta^*_b)$, then as data accumulates, the posterior $W_1$ will, with high probability, concentrate narrowly around $(\theta^*_a,\theta^*_b)$ and so 
$S_{(j),[n_a,n_b,W_1]}$
will behave more and more similarly to the `best' $(\theta_a^*,\theta_b^*)$ \E variable.  Still, with the exception of a special case we indicate below, in general we cannot expect it to be the $W_1$-GRO E-variable. But we are not particularly concerned by this: our experiments in Section~\ref{sec:experiments} indicate that, at least in  the $2 \times 2$ table setting, it behaves quite well in terms of power, which is often the main practical interest. 
\paragraph{Simplification when $\{P_{\theta}: \theta \in \Theta \}$ is Convex and $\cY$ is finite}
Denoting $W_{1,g}|Y^{(m)}$ as the marginal posterior for $\theta_g$, for $g \in \{a,b\}$, we can rewrite (\ref{eq:sixth}) as 
\begin{align}\label{eq:sixthb}
S^{(m)}_{[n_a,n_b,W_1]}  = 
& \prod_{j=1}^m  
\frac{\prod_{i=1}^{n_a} {p_{W_{1,a}|Y^{(j-1)}}(Y_{(j-1)n_a +i,a})}
\prod_{i=1}^{n_b} 
{p_{W_{1,b}|Y^{(j-1)}}(Y_{(j-1)n_b +i,b})}}
{
\prod_{g \in \{a,b\}}
\prod_{i=1}^{n_g} \left( \frac{n_a}{n} p_{W_{1,a}|Y^{(j-1)}}(Y_{(j-1) n_g +i,g})
+ \frac{n_b}{n} p_{W_{1,b}|Y^{(j-1)}}(Y_{(j-1) n_g +i,g})
\right) 
} \nonumber \\
\overset{\text{if $\{P_{\theta}: \theta \in \Theta \}$ convex, $\cY$ finite}}{=}\  &   \prod_{j=1}^m  
\prod_{i=1}^{n_a} \frac{p_{W_{1,a}|Y^{(j-1)}}(Y_{(j-1)n_a +i,a})}{
p_{\breve\theta_0 |Y^{(j-1)}}(Y_{(j-1) n_a +i,a})}
\prod_{i=1}^{n_b} 
\frac{p_{W_{1,b}|Y^{(j-1)}}(Y_{(j-1) n_b +i,b})}{
p_{\breve\theta_0 |Y^{(j-1)}}(Y_{(j-1) n_b +i,b})
}
\end{align}
with 
$\breve\theta_0 |Y^{(j-1)} \in \Theta \text{\ s.t.\ }
p_{\breve\theta_0 |Y^{(j-1)}} = 
 (n_a/n) p_{W_{1,a}| Y^{(j-1)}} + (n_b/n) p_{W_{1,b}|Y^{(j-1)}}$, the existence of $\breve\theta_0|Y^{(j-1)}$ being guaranteed if $\{P_{\theta}: \theta \in \Theta \}$ is convex and the sample space is finite (for then, by Carath\'eodory's Theorem, for any distribution $W$ on $\Theta$ there is a distribution $W'$ on $\Theta$ with finite support such that $p_W= p_{W'}$, and by convexity, there is $\theta^{\circ}$ such that $p_{W'} = p_{\theta^{\circ}}$). This rewrite will enable several additional results for such $\Theta$.

\paragraph{Connection to Bayes Factors}
Consider $W_1$ such that $\theta_a$ and $\theta_b$ are independent under $W_1$ with marginal distributions $W_a$ and $W_b$, and now further take $n_a = n_b = 1$. By basic telescoping, and using that, if independent under the prior, $\theta_a$ and $\theta_b$ must also be independent under the posterior,  
we can then further rewrite (\ref{eq:sixth}) as 
\begin{align}\label{eq:sixthc}
& \frac{
\int 
p_{\theta_a}(Y^m_a) dW_a(\theta_a) \int  p_{\theta_b}(Y^m_b)
d W_b(\theta_b)
}{ \prod_{j=1}^m\prod_{g\in \{a,b\}}
 \left( \frac{1}{2} p_{W_{1,a}| Y^{(j-1)}} (Y_{j,g}) + \frac{1}{2} p_{W_{1,b}|Y^{(j-1)}} (Y_{j,g}) \right)
} \ \overset{\text{if $\{P_{\theta}: \theta \in \Theta \}$ convex}}{=}  \\ 
\label{eq:seventh}
& 
\frac{
\int 
p_{\theta_a}(Y^m_a) dW_a(\theta_a) \int  p_{\theta_b}(Y^m_b)
 d W_b(\theta_b)
}{ \prod_{j=1}^m \prod_{g\in \{a,b\}} p_{\breve\theta_0|Y^{(j-1)}}(Y_{j,g})
}
\end{align} 
where the equality holds if $\{P_{\theta}: \theta \in \Theta_0\}$ is convex and $\cY$ is finite so that (\ref{eq:sixthb}) holds.
As seen from (\ref{eq:sixthc}), even without finiteness or convexity, the numerator of the generic product \E -value is now equal to the Bayesian marginal likelihood of the data based on prior $W_1$. 
\commentout{While in general, this rewrite of the numerator breaks down if $n_a$ or $n_b$ is larger than one and/or $\theta_a$ and $\theta_b$ are dependent under the prior, 
The reason is that according to (\ref{eq:sixth}) (see the line above (\ref{eq:sixthb})), the outcomes in each group within each of the $m$ blocks are treated as independent (they have the same density).  
In contrast, while a Bayes factor for $m$ blocks of $n$ data points can also be rewritten as a product of $n \cdot m$ Bayes predictive densities, in general it makes every outcome dependent on every previous outcome --- the posterior, and hence the posterior predictive, is updated also within each block. If $\theta_a$ and $\theta_b$ are independent under the prior though, then they are also independent under the posterior, and if $n_a = n_b = 1$ then the Bayes predictive densities for the two outcomes within each block will also be independent, and we get (\ref{eq:seventh}).
}
Thus, in this special case (i.e. $n_a=n_b=1$, prior independence; the derivation breaks down if these do not hold), if the denominator could also be written as a Bayes marginal likelihood, then our \E variable would really be a Bayes factor. Yet, even if $\{P_{\theta}: \theta \in \Theta\}$ is convex, it cannot be written in this way, though it is very `close': each of the $m$ factors in the denominator in (\ref{eq:seventh}) is the product density function of two identical distributions for one outcome, and Proposition~\ref{prop:beta} shows that, in the special case of the $2 \times 2$ model with $W_a$ and $W_b$ independent beta priors, this distribution may itself be the Bayes predictive distribution obtained by equipping $\Theta_0$ with another beta prior. Still, 
for a real Bayes factor corresponding to $\cH_0$, for each $j$, the two outcomes $Y_{j,a}, Y_{j,b}$ in the $j$-th block would not be independent given $Y^{(j-1)}$, whereas in (\ref{eq:seventh}) they are, so we may conclude that in general, our e-variables are not equivalent to any Bayes factor.

\section{Safe tests for Two Proportions}
\label{proportiontestS}
We assume the setting above and, for now, assume that both streams are Bernoulli. This will substantially simplify the formulae. Thus, $\Theta =[0,1]$ and  
(\ref{eq:firstgeneral}) now specializes to
\begin{equation} \label{eq:first}
p_{\theta_a,\theta_b} (y_{a}^{t_a}, y^{t_{b}}_b) \coloneqq 
p_{\theta_a} (y_{1,a}, \ldots, y_{t_a,a}) p_{\theta_b}(y_{1,b}, \ldots, y_{t_b,b}) = 
\theta_a^{t_{a1}}(1-\theta_a)^{t_a - t_{a1}}\theta_b^{t_{b1}}(1-\theta_b)^{t_b - t_{b1}}.
\end{equation}
with 
$t_{a1}$ the number of outcomes $1$ in stream $a$ among the first $t_a$ ones, and $t_{b1}$ the number of outcomes $1$ in stream $b$ among the first $t_b$ ones.
According to the {null hypothesis},  we have that $\theta^*_a = \theta^*_b = \theta_0$ for some $\theta_0 \in  \Theta=[0,1]$.
(\ref{eq:first}) now simplifies to: 
\begin{equation*} 
p_{\theta_0}(y_a^{t_a},y_b^{t_b}) \coloneqq  \theta_0^{t_1}(1-\theta_0)^{t_0},
\end{equation*}
with $t_1$ the number of ones in the sequence $y^{t_a+t_b}= y_{1}, \ldots, y_{t_{a}+t_b}$, and similarly for $t_0$.

We now run through the results of the previous section for this instantiation of our test. Again, we start with the case of a simple $\cH_1 = \{P_{\theta^*_a,\theta^*_b} \}$. 
(\ref{eq:simpleevar}) can now be written as:  
\begin{equation}\label{eq:simpleevarb}
s(y^{n_a}_a, y^{n_b}_b; n_a, n_b, \theta^*_a,\theta^*_b) 
\coloneqq  \frac{p_{\theta_{a}^*}( y^{n_a}_a) }{p_{\theta_0}( y^{n_a}_a)} \cdot 
 \frac{p_{\theta_b^*}(y^{n_b}_b)}{p_{\theta_0}(y^{n_b}_b)
}, \ \ \text{where}\ \ \theta_0 = \frac{n_a}{n} \theta^*_a + \frac{n_b}{n} \theta^*_b.
 \end{equation}
Theorem~\ref{simpleSproportions} tells us that this is an \E variable. Since $\{P_{\theta}: \theta \in \Theta \}$, the Bernoulli model,  is convex, the theorem also tells us that in this case  the generic \E variable with simple alternative is always $(\theta^*_a,\theta^*_b)$-GRO.

We now turn to the generic \E-variable relative to arbitrary prior $W_1$. For the Bernoulli model the Bayes posterior predictive distribution is itself a Bernoulli distribution, with its parameter equal to the posterior mean. Therefore, while the generic \E-variable relative to prior $W_1$ is still given by (\ref{eq:simplewithprior}), this now simplifies to:
\begin{equation}\label{eq:twobytwoevar}
s(y^{n_a}_a, y^{n_b}_b; n_a, n_b, W_1) =  s(y^{n_a}_a, y^{n_b}_b; n_a,n_b, \theta^*_a,\theta^*_b) \text{\ for $\theta^*_g = {\bf E}_{\theta_g \sim W_1}[\theta_g], g \in \{a,b\}$.}
\end{equation}
Combining  this with (\ref{eq:sixthb}) we infer that 
\begin{align}\label{eq:niceeval}
S^{(m)}_{[n_a,n_b,W_1]} & =  
\prod_{j=1}^m  
\prod_{i=1}^{n_a} \frac{p_{\breve\theta_a|Y^{(j-1)}}(Y_{(j-1)n_a +i,a})}{
p_{\breve\theta_0 |Y^{(j-1)}}(Y_{(j-1) n_a +i,a})}
\prod_{i=1}^{n_b} 
\frac{p_{\breve\theta_b|Y^{(j-1)}}(Y_{(j-1) n_b +i,b})}{
p_{\breve\theta_0 |Y^{(j-1)}}(Y_{(j-1) n_b +i,b})
}
\end{align}
where  $
\breve\theta_a | Y^{(j-1)} = {\bf E}_{\theta_a \sim W \mid Y^{(j-1)}}[\theta_a]$ and $\breve\theta_b |Y^{(j-1)} = {\bf E}_{\theta_b \sim W  \mid Y^{(j-1)}}[\theta_b]$ and  $\breve\theta_0 |Y^{(j-1)} = (n_a/n)
\breve\theta_a \mid Y^{(j-1)} + (n_b/n)\breve\theta_b \mid Y^{(j-1)}$.
\paragraph{Simplified Calculations with Independent Beta Priors}
Now take the special case in which $\theta_a$ and $\theta_b$ are independent under the prior $W_1$ with marginals $W_a$ and $W_b$.
In this case, $\theta_a$ and $\theta_b$ are also independent under the posterior, and we can simplify 
$
\breve\theta_a | Y^{(j-1)} = {\bf E}_{\theta_a \sim W_a \mid Y^{(j-1)n_a}_{a}}[\theta_a]$, the expectation  of $\theta_a$ under the posterior $W_a$ given all data so far in group $a$, and similarly for group $b$. Using beta priors, this expectation is easy to calculate and we get:
\begin{proposition}\label{prop:beta}
Let $\theta_a,\theta_b$ be independent under $W_1$, with  marginals $W_a$ and $W_b$ respectively. Suppose that these are beta priors with parameters
$(\alpha_a,\beta_a)$ and $(\alpha_b,\beta_b)$ respectively. Then, upon defining
$U_a = \sum_{i=1}^{(j-1) n_a} Y_{i,a}, U_b = \sum_{i=1}^{(j-1) n_b} Y_{i,b},
U = \sum_{i=1}^{(j-1) n} (Y_{i,a} + Y_{i,b})$ we have that 
$\breve\theta_a,\breve\theta_b,\breve\theta_0$ as above satisfy: $\breve\theta_a| Y^{(j-1)}  = (U_a + \alpha_a)/((j-1) n_a + \alpha_a + \beta_a )$, $\breve\theta_b| Y^{(j-1)}  = 
 (U_b + \alpha_b)/((j-1) n_b + \alpha_b + \beta_b )
$ respectively, and $\breve\theta_0| Y^{(j-1)}$ is as further above. 
In the special case that we fix the prior parameters in the groups proportional to the group size fraction $\kappa:= n_b/n_a$, i.e we fix $\alpha_b = \kappa \alpha_a$, $\beta_b = \kappa \beta_a$, the expression for $\breve\theta_0$ simplifies to  $\breve\theta_0 |Y^{(j-1)} = 
(U + (1+\kappa) \alpha_a)/
((j-1) n + (1 +\kappa) \alpha_a + (1+ \kappa) \beta_a)
$. 
\end{proposition}
\section{(Un)Restricted Composite \hmc{1}  in the $2 \times 2$ setting}\label{sec:compositeH1}
In this section we describe the main instantiations of the $2 \times 2$ stream testing scenario that are relevant in practice. These differ in the choice of $\cH_1$: the choice can be fully unrestricted (we simply want to find whether there is any discrepancy from $\cH_0$ at all); restricted in terms of effect size; or restricted because we have prior knowledge about either $\theta^*_a$ or $\theta^*_b$.
We consider each in turn, the second and third scenario in a separate subsection.  Section~\ref{sec:experiments} provides extensive numerical simulations for all three scenarios. 

In the first scenario, a researcher wants to perform a \emph{two-sided test}; they simply aim to find any discrepancy from \hmc{0} if it exists, with no restrictions are placed on \hmc{1}. In this case, if we choose $W_1$ as independent beta priors on $\theta_a$ and $\theta_b$, we can simply proceed as described in Proposition \ref{prop:beta} above, taking a beta prior for simplicity. We will develop a reasonable `default' choice for the hyper parameters by experiment in Section~\ref{sec:experiments}. 
\subsection{Dealing with Effect Sizes}
\label{sec:effectsize}
In the second scenario we really want to test $\cH_0$ against a restricted $\cH_1$ consisting of those hypotheses that have a certain minimal {\em effect size\/} $\delta$. This would then be a one-sided test. For example, a researcher might know that a new treatment must cure at least a certain number of patients more compared to a control treatment to provide \emph{a clinically relevant treatment effect} $\delta$. 
In this case, $\cH_1$ could be restricted to either of the sets $\Theta(\delta)$ or  $\Theta^+(\delta)$, where
\begin{equation}
\label{eq parameter space absolute}
    \Theta(\delta) = \left\{\theta \in [0,1]^2:  d(\theta) = \delta \right\} \ \ ; \ \ 
    \Theta^+(\delta) = \begin{cases}
    \left\{\theta \in [0,1]^2:  d(\theta) \geq  \delta \right\} & \text{if $\delta > 0$}\\ 
 \left\{\theta \in [0,1]^2:  d(\theta) \leq  \delta \right\} & \text{if $\delta < 0$,}
\end{cases} \end{equation}
where  we set $d((\theta_a, \theta_b)) = \theta_b- \theta_a$. 
A second notion of effect size that often will be applicable in this sort of research is the  \emph{log odds ratio} between $\theta_b$ and $\theta_a$, with restricted parameter space again given by (\ref{eq parameter space absolute}) but $d$ set to 
\begin{equation}
\label{eq:logodds}
d( (\theta_a,\theta_b)) = \log\left(\frac{\theta_b}{1-\theta_b} \cdot \frac{1-\theta_a}{\theta_a}\right).
\end{equation}
These are the two effect size notions that will feature in our experiments. 
An illustration of both divergence measures and the resulting restricted parameter spaces is given in Figure \ref{fig:parameterSpaceExamples}. 

A third popular notion of effect size, the relative risk, behaves,  for small $\theta_a$ and $\delta > 0$, very similarly to the odds ratio, and will therefore not be separately considered in our experiments. 
\commentout{We will also not consider two-sided testing against a restricted \hmc{1} (i.e., one wants to restrict the alternative hypothesis to a treatment being either substantially better \emph{or} substantially worse than a control) since this is not a very common scenario. Such \E variables for the $2 \times 2$ setting and stream data could however be constructed with a method analogous to the one we describe below, by combining two ``simple'' \E variables \citep{Turner19}.
}

\begin{figure}[ht]
     \centering
     \begin{subfigure}[b]{7cm}
         \centering
         \includegraphics[width=\textwidth]{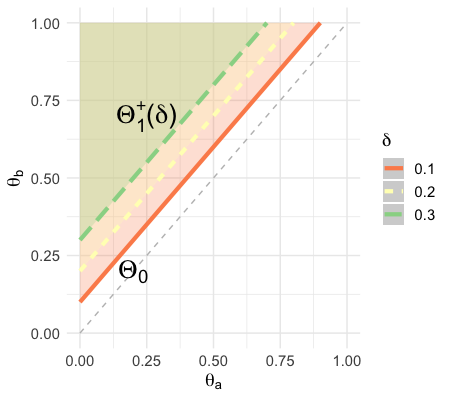}
         \caption{$d((\theta_a, \theta_b)) = \theta_b- \theta_a$}
     \end{subfigure}
     \hfill
     \begin{subfigure}[b]{7cm}
         \centering
         \includegraphics[width=\textwidth]{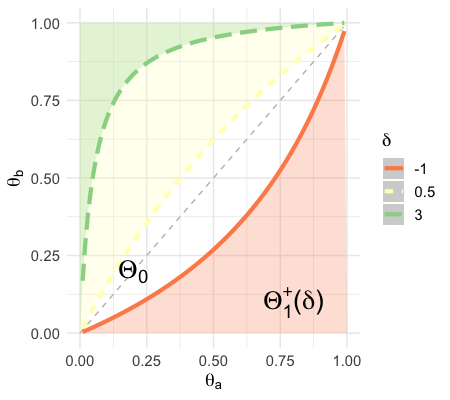}
         \caption{$d( (\theta_b,\theta_a)) = \log\left[\frac{\theta_b}{1-\theta_b}\frac{1-\theta_a}{\theta_a}\right]$}
     \end{subfigure}
        \caption{Examples of restricted alternative hypothesis parameter spaces for several values of two divergence measures; the difference between group means and the log odds ratio. $\Theta_0$ denotes the null hypothesis parameter space; $\Theta_1^+(\delta)$ the restricted alternative hypothesis parameter space.}
        \label{fig:parameterSpaceExamples}
\end{figure}

\commentout{
\paragraph{A (sub)optimal \E variable for restricted \hmc{1}} What would be the ideal \E -variable for $\cH_1$ thus restricted? We first consider the case of a $\cH_1 = \{P_{\theta}: \theta \in \Theta(\delta')\}$ with a precise effect size $\delta' \neq 0$. For this case, following the motivation of Section~\ref{sec:basics}, we simply adopt the criterion (\ref{eq:regret}), re-stated for convenience below:
\begin{align}\label{eq:regretc}
\textsc{regret}(\trv{S}{m})  := \max_{\theta \in \Theta(\delta')} {\bf E}_{\trv{Y}{m} \sim P_{\theta}} [ \log S^{(m)}_{\gro(\theta)} - \log \trv{S}{m}{} ],
\end{align}
The ideal \E variable at sample block size $m$ would be the \E variable achieving, among all \E variables for $\cH_0$, the minimum regret (\ref{eq:regretc}). We will refrain from the goal of optimizing (\ref{eq:regretc}) precisely. Instead we will consider \E variable processes as defined in \eqref{eq:niceeval} with priors $W_1$ on $\Theta(\delta')$: $S^{(m)}_{n_a, n_b, W_1}$. Among these, we will search for the prior that approximately minimizes (\ref{eq:regretc}). Since the optimizing prior depends on $m$, we will design as our default choice a prior $W^*_1$ that is suboptimal (but not too bad) for small $m$, and optimal, among the class considered, for large values of $m$ in Section \ref{sec:bestpriors}.}

If we pick $\cH_1$ restrict to $\Theta(\delta')$, then we could simply use the beta prior mentioned before with support conditioned on this set. What about the more realistic case of a $\cH_1$ with $\delta \in \Theta^+(\delta')$? A first, intuitive (and certainly defensible) approach would be to use a prior $W'_1$ that is spread out over $\Theta^+(\delta')$, e.g. (if $\delta'> 0$)  the beta prior as above conditioned on $\delta \geq \delta$. However, in terms of the GRO criterion, there are good reasons to still use a prior $W^*_1$ that puts all prior mass on $\Theta(\delta')$, the boundary of the real parameter space $\Theta(\delta^+)$. Namely, for the resulting \E variable process $S^{(1)}_{[n_a, n_b, W^*_1]}, S^{(2)}_{[n_a, n_b, W^*_1]}, \ldots$, it holds for every $m$ that 
\begin{multline}\label{eq:regretd}
\text{for all $(\theta_a, \theta_b)$ with $d((\theta_a, \theta_b)) > \delta'$}, \  
{\bf E}_{\trv{Y}{m} \sim P_{(\theta_a, \theta_b)}} [  \log S^{(m)}_{[n_a, n_b, W^*_1]} ] \geq \\
\min_{\theta \in \Theta(\delta')} {\bf E}_{\trv{Y}{m} \sim P_{\theta}} [  \log S^{(m)}_{[n_a, n_b, W^*_1]} ].
\end{multline} 
Thus, we might want to use the  prior $W^*_1$ also if $\delta$ can be more extreme than $\delta'$,  since 
if $\delta$ is actually more extreme, the expected (log-) evidence against $\cH_0$ using $W_1^*$ (even though designed for $\delta'$) will actually get larger anyway. 

The advantage of the first approach is that it will lead to much higher GROwth 
($
{\bf E}_{P_{(\theta_a, \theta_b)}} [  \log S^{(m)}_{[n_a, n_b, W'_1]} ]$
much larger than ${\bf E}_{P_{(\theta_a, \theta_b)}} [  \log S^{(m)}_{[n_a, n_b, W^*_1]} ]$)
if we are `lucky' and
$|d(\theta_a, \theta_b)| \gg  |\delta'|$. The price to pay is that it will lead to somewhat smaller growth if $d((\theta_a, \theta_b))$ is (still arger than but) close to $\delta'$ (experiments omitted). It is easy to see why: the prior $W'_1$ must spread out its mass over a much larger subset of $[0,1]^2$ than $W^*_1$. Therefore, the E-variables based on $W'_1$ will perform somewhat worse than those based on $W^*_1$ if the data are sampled from a point $(\theta^*_a, \theta_b^*)$ in  the support of $W^*_1$, simply because $W^*_1$ gives much larger prior support in a neighborhood of $(\theta^*_a, \theta_b^*)$.
For this reason, and also because it is computationally a lot simpler, we decided to focus our experiments on the second approach rather than the first. 

\paragraph{Calculating the prior and posterior for restricted \hmc{1}}
For both notions of effect size, $\theta_a$ and $\theta_b$ can no longer be independent for any prior on $\Theta(\delta)$. Hence, the prior and posterior do not longer admit the composition in terms of  beta densities as in Proposition~\ref{prop:beta}. For example, when putting a prior on $\Theta(\delta)$ with the additive effect size notion, we know the new domain of $\theta_a$ would be $[0, 1 - \delta]$. $\theta_b$ is completely determined by $\theta_a$ and $\delta$ in this case. We will still use a beta prior on $\Theta(\delta)$ and calculate posteriors by
a numerical approach, explained in Appendix~\ref{app:numerical} of the Supplementary Material. 

\subsection{Working with Restrictions on event rate}
\label{subsection: prior knowledge}
In practice, researchers often already have estimates of the occurrence rate of events in the {control group} in their experiments; for example, estimates of the proportion of patients that recover from a disease under standard care are known, and researchers investigate whether the proportion of recovered patients is higher in a group receiving an experimental treatment. This restriction on $\theta_a$ can be incorporated in the \E variable. This incorporation becomes especially easy if $\cH_1$ is already restricted to a set $\Theta^+(\delta')$ with minimal relevant effect size $\delta'$. For then $\Theta(\delta')$ contains just one point $(\theta_a^*, \theta_b^*)$ (in the case of the linear effect size, this is $(\theta_a,\theta_a+\delta)$), and the \E-variable constructed according to the guidelines of the previous subsection, which puts all its mass on $\delta'$ even though we allow $\delta \geq \delta'$, would be the generic \E-variable corresponding to putting prior mass 1 on $(\theta_a^*, \theta_b^*)$. 


\section{Illustration via Simulated Data}
\label{sec:experiments}
In this section, we illustrate properties of our \E variables for $2 \times 2$ application through simulated data, generated with our software package publicly available through Github \citep{LyT20}.
First, we determine a reasonable choice of beta prior hyper-parameter to use in \eqref{eq:niceeval} in terms of the GRO-criterion. Thereafter, we show by more simulations that our proposal for the beta prior hyper-parameter based on GRO also performs well in terms of power (recall from Section~\ref{sec:basics} that while we cannot optimize for power directly, we do want procedures with reasonable power). Finally, we compare the power of our \E variable with this default prior choice and different restrictions on \hmc{1} to Fisher's exact test. 

\paragraph{REGROW}
For simplicity, in all our experiments we will invariably set the beta prior hyper-parameters to  $\alpha_a=\alpha_b=\beta_a=\beta_b= \gamma$ for some $\gamma > 0$ (recall that any such choice leads to a valid \E variable).  We will aim for the $\gamma$ that minimizes (\ref{eq:jasaregret}) in the worst-case over all $\theta^*_1 \in [0,1]^2$, thereby following the REGROW {\em ({\em relative\/} growth-rate optimality in worst-case)\/} criterion of \cite{grunwald2019safe}, who give a minimax regret motivation for this choice. In essence, the prior minimizing, among all distributions over $[0,1]^2$, the maximum of (\ref{eq:jasaregret}) over all $\theta^*_1$ can be viewed as the prior that allows us to learn $\theta^*_1$ as fast as possible (based on a minimal sample) in the worst-case. Here we are contented to adopt a sub-optimal but computationally convenient prior by restricting the minimum to be over a 1-dimensional family of beta priors with hyper parameter $\gamma$. We find the minimizing $\gamma$ by experiment: results are depicted in Figure \ref{fig:regret}. It 
depends on $m$, which is unknown in advance, but for large $m$, in the setting with $n_a = n_b = 1$, it converges to $\gamma \approx 0.18$, and this is the value we will take as our default choice --- our experiments below indicate that it remains a good choice, also when our main concern is power, and also under restrictions on $\cH_1$.

\begin{figure}[ht]
    \centering
    \begin{subfigure}{8cm}
     \centering
     \includegraphics[width=\textwidth]{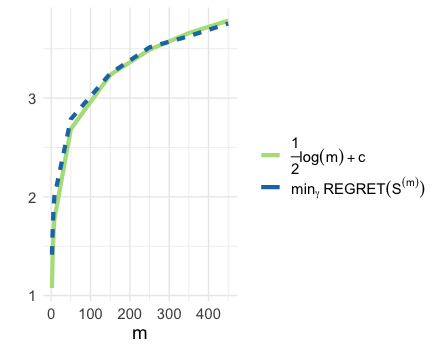}
     \caption{$\min_{\gamma} \regret S^{(m)}$}
    \end{subfigure}
    \hfill
    \begin{subfigure}{8cm}
          \centering
     \includegraphics[width=\textwidth]{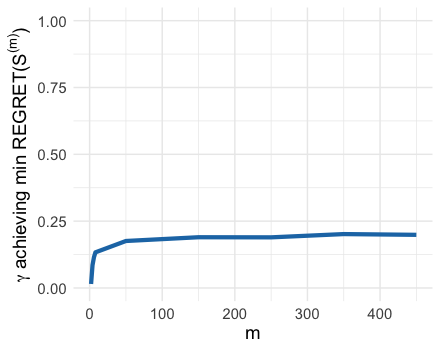}
     \caption{$\argmin_{\gamma} \regret S^{(m)}$}
    \end{subfigure}
    \caption{Minimized regret w.r.t. Beta prior hyperparameter $\gamma$ for the two-sample stream \E variable for two proporions \eqref{eq:twobytwoevar}. Relative growth rate (see \eqref{eq:jasaregret}) was estimated through 10000 simulations and $\regret$ was calculated as the maximum over $\theta_1^*.$}
    \label{fig:regret}
\end{figure}

\paragraph{Power}
Whereas GROwth is the natural performance measure in experiments that may always be continued at some point in the future, traditionally oriented researchers may be more interested in power. The question is then whether the optimal asymptotic choice $\gamma \approx 0.18$ in terms of the relative GRO property for unrestricted $\cH_1$ is also the optimal choice in terms of power (which is usually considered in combination with some minimal effect size, i.e. a restricted $\cH_1$). The following experiment shows that by and large it is. For simplicity we only illustrate the case $n_a = n_b = 1$ and a desired power of $0.8$. For various effect sizes $\delta$, and various values of $\gamma$, we first determined the smallest sample size (number of blocks) $m$ such that, under optional stopping up until and including $m$, the power is $\geq 0.8$ in the worst case over all $(\theta_a,\theta_b)$ with $\delta = \theta_b - \theta_a$. Here by `optional stopping up until and including $m$', we mean `we stop and reject the null  iff $S^{(m')}_{[n_a,n_b,W_{[\gamma]}} > \alpha^{-1}$ for some $m' \in \{1,2, \ldots, m$\}, and we stop and accept the null if this is not the case (so $m$ is the maximal sample size we consider)'.  
We call this $m$ the {\em worst-case\/} sample size needed for $80\%$ power at effect size $\delta$ with prior parameter $\gamma$.  
The reason for calling it worst-case is that in practice, by engaging in optional stopping with a fixed maximal sample size, the  \emph{expected sample size} of this procedure is smaller: if, for $m'< m$, we already have $S^{(m')}_{[n_a,n_b,W_{[\gamma]}} > \alpha^{-1}$ then we stop and reject early; if not, we go on until we have seen $m$ blocks and then stop (and reject iff $S^{(m)}_{[n_a,n_b,W_{[\gamma]}} > \alpha^{-1}$). 
We thus performed two simulation experiments: first, to estimate the worst-case sample size (at $\alpha = 0.05$), and second, to estimate the expected sample size. Again, the estimates were obtained by re-simulating a sequence of data blocks $K$ times for a large number of $K$, making sure the bias and variance of the estimates were sufficiently small. 
\begin{figure}[ht]
    \centering
     \begin{subfigure}[b]{10.4cm}
     \centering
         \includegraphics[width =\textwidth]{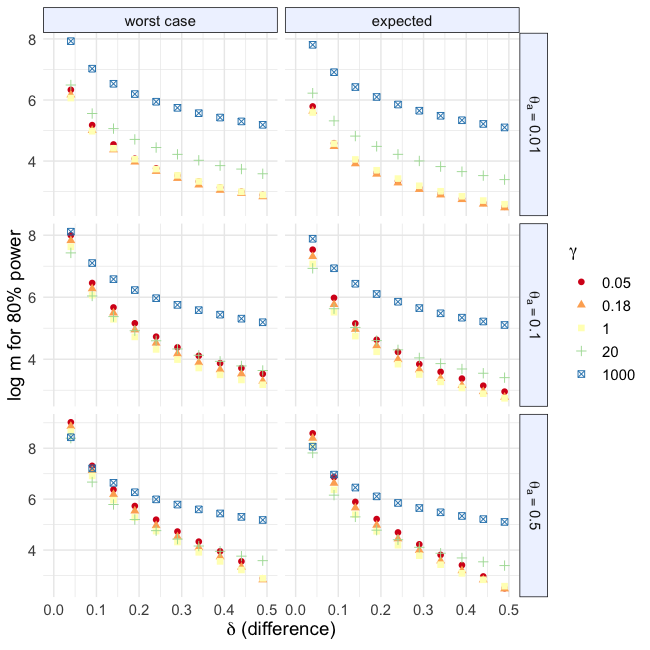}
     \caption{$m$ on log scale}
    \end{subfigure}
         \begin{subfigure}[b]{4cm}
     \centering
         \includegraphics[width =\textwidth]{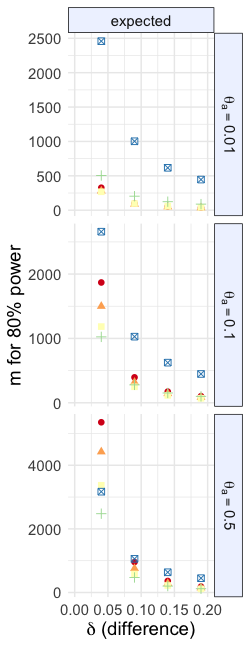}
     \caption{$m$, identity}
    \end{subfigure}
    \caption{In 2000 simulations, the (natural logarithm of) the number of data blocks $m$ (``sample sizes") needed for achieving $80\%$ power while testing at $\alpha = 0.05$ for distributions with varying group means and varying differences between group means were estimated for different beta prior parameter values.
    }
    \label{fig:powermisspecification}
\end{figure}
In Figure \ref{fig:powermisspecification} results of these experiments are depicted. We make two observations: first, almost no difference in sample sizes to plan for between $\gamma= 0.18$ and $\gamma = 0.05$ was observed for distributions with small expected sample sizes (represented by the triangles and the dots, which overlap for most data points), and other values of $\gamma$ obtained smaller power, indicating that the relative growth-optimal $\gamma = 0.18$ could in practice be used as a default setting for our \E variable --- and as a consequence, we recommend it as such. Second, in the rightmost panel we see that for distributions with {\em very\/} small relative differences between $\theta_a$ and $\theta_b$, e.g. $P_{0.5, 0.58}$, values of $\gamma$ higher than $0.18$ yielded a higher power, whereas for such $\delta$, the relative GROW criterion was optimized for $\gamma = 0.18$ for the corresponding (very large) stopping times in our simulation experiments.
This is not surprising given what is known for simple $\cH_0 = \{P_{\theta_0}\}$: when testing a point null $\theta_0$ with a 1-dimensional exponential family alternative, safe tests based on Bayes factors with standard Bayesian (e.g. Gaussian or conjugate) priors do not  obtain optimal power in an asymptotic sense: they reject if $|\hat\theta - \theta_0|^2 \gtrsim (\log n)/n$ (with $\hat\theta$ denoting the MLE; see the example on $Z$-tests by \cite{grunwald2019safe}) whereas based on nonstandard `switching' \citep{PasG18} or `stitching' methods \citep{howard2018uniform}, corresponding to special priors with densities going to infinity as effect size goes to $0$, one can get rejection if $|\hat\theta - \theta_0|^2 \gtrsim (\log \log n)/n$. However, there is a significant price to pay in terms of the constants hidden in the asymptotics, and in practice, `standard' priors may very well perform better at all but very large sample sizes  \citep{maillard2019mathematics}. Given that the higher $\gamma$, the more the beta prior behaves like a switch prior, we conjecture that what  we see in Figure~\ref{fig:powermisspecification}(b) at very small $\delta$ is a version of the switching/stitching phenomenon with a composite null; since it only kicks in at very large sample sizes, we prefer $\gamma = 0.18$ as the default choice after all. 

Finally, we compared the performance of our \E variables with the ``default" beta priors with $\gamma = 0.18$ with their classical counterpart, Fisher's exact test. We show that with Fisher's exact test, type-I error probability guarantee is lost, whereas with the \E variables it remains bounded --- since these results are exactly as would be expected from the theory they have been placed in the supplementary material (Figure \ref{fig:type1Error} in Appendix \ref{app:experiments} in the Supporting Material). In the main text below, we compare worst-case and expected stopping times of the \E variables with- and without restrictions on \hmc{1} for sample sizes one would need to plan for when analyzing experiment results with Fisher's exact test; see Figure \ref{fig:powerfisher}. We noticed that the expected sample sizes achieved under optional stopping with the \E variable with unrestricted \hmc{1} were very similar to the sample sizes needed to plan for with Fisher's exact test. When using a correctly specified restriction on \hmc{1} (the leftmost data points in the second and third subfigures), this expected number of samples is even considerably lower than the sample size to plan for with Fisher's exact test. However, under misspecification, when the difference or log odds ratio used in the design of the \E variable turns out to be a lot smaller than the real difference present in the data generating machinery, one should expect to collect more samples (the data points towards the right in the second subfigure). This effect would disappear if we were to put a prior on the full $\Theta^+(\delta)$ rather than the boundary $\Theta(\delta)$, at the price of slightly worse behaviour in the well-specified case when data is sampled from $\Theta(\delta)$.

Note that in Figure~\ref{fig:powerfisher} we used the default beta prior parameters $\gamma = 0.18$ found optimal for the unrestricted case for the restricted cases as well; some first experiments revealed that changing the prior parameter values did not lead to significant changes in power for the restricted \E variables (results not shown). We do however offer the possibility in our software package \citep{LyT20} to run similar experiments for users to determine the optimal prior parameter $\gamma$  for a given expected sample size and $\Theta^{(+)}(\delta')$.

\begin{figure}[ht]
    \centering
    \includegraphics[width = 12cm]{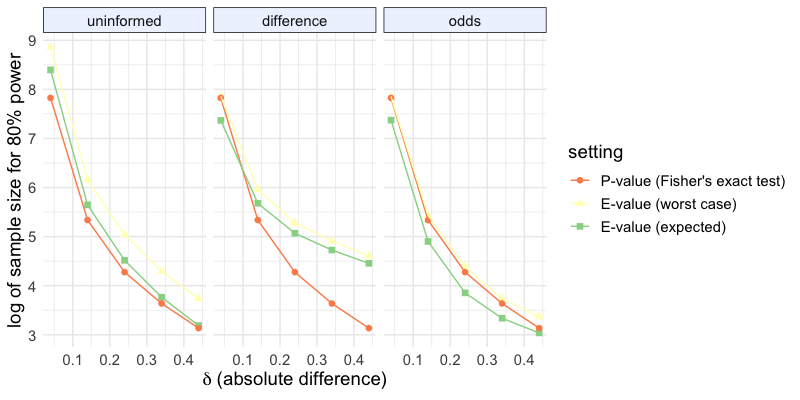}
    \caption{Estimates from 1000 simulations of worst-case and expected sample sizes for achieving 80\% power estimated for three types of \E variables with different restrictions on \hmc{1}, and the sample size to plan for with Fisher's exact test. Hypothesized effect sizes were $0.04$ for the \E variables with prior information on the absolute difference and were converted equivalently for the log odds ratio prior information case, and we set $\gamma = 0.18$ for the beta priors.}
    \label{fig:powerfisher}
\end{figure}

\paragraph{Beyond Two-Stream Data: Safe Tests for $K$ Proportions} We also compared the performance of the extended version of our \E variable for $k$ Bernoulli data streams to the corresponding classical, nonsequential counterpart, the chi-square test \citep{mchugh2013chi}. In this setting, we have a $k \times 2$ contingency table test, where we test whether $k$ Bernoulli data streams come from the same source. The extension of \eqref{eq:niceeval} to $k$ data streams analogously to \eqref{eq:evarkextension} is straightforward. Our \E variable with uniform priors significantly outperforms the chi-square test for small sample sizes and large effect sizes (see Figure \ref{fig:powerchisq}), probably explained by the fact that the chi-square test is not exact, but the \E variable is. For expected cell counts smaller than $5$ the chi-square test should not be used, reflected in an increased number of samples needed for similar power \citep{mchugh2013chi}.

\begin{figure}[ht]
    \centering
    \includegraphics[width = 6.75cm]{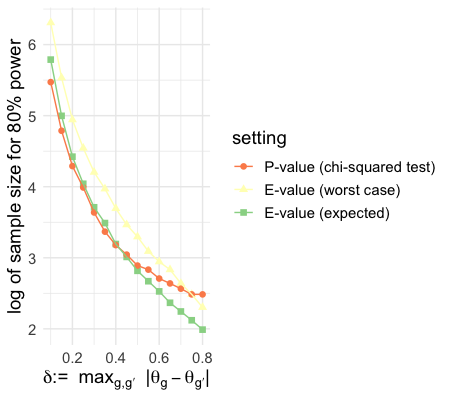}
    \caption{Estimates from 1000 simulations of worst-case and expected sample sizes for achieving 80\% power estimated for testing with the $k$-stream \E variable, and the sample size to plan for with the chi-square test. Data were simulated with balanced data blocks, $\vec{n} = (1,1,1,1)$ and $\vec{\theta}$ was set as an equally spaced grid from $\theta_a = 0.1$ to $\theta_k = \theta_a + \delta_{\text{max}}$. We set $\gamma = 1$ for the beta priors. We see that for large enough $\delta_{\max}$, the expected sample size becomes significantly smaller than the fixed sample size needed for the chi-squared test, overtaking it at approximately $\delta_{\max} =0.45$, which really means that one third of the times the effect size is $0.15$, one third it is $0.3$ and one third $0.45$. At these $\delta$, the expected sample size for our $2$-stream \E variable is still larger than the fixed number needed for Fisher's exact test.}
    \label{fig:powerchisq}
\end{figure}
\section{Illustration via Real World Data}
\label{sec:realworld}
We will now demonstrate the approach through a real-world example: the SWEPIS study on labor induction \citep{wennerholm2019induction}. \cite{wagenmakers_ly_2020} have used this example before to illustrate how using single p-values to make decisions can hide valuable information in research data.

In the SWEPIS study, two groups of pregnant women were followed. In the first group labor was induced at 41 weeks, and in the second labor was induced after 42 weeks. The study was stopped early, as 6 cases of stillbirth were observed in the 42-weeks group (at $n_b = 1379$), as compared to 0 in the 41-weeks group (at $n_a = 1381$). These data yield a significant Fisher's exact test, $P \approx 0.015$, for testing that the number of stillbirths in the 42-weeks group is higher, when (wrongly) assuming that $n_a$ and $n_b$ were fixed in advance to the above values.

If we had used \E variables for continuously analyzing this data, would we then have found evidence for superiority of the 41 weeks approach, and would we have stopped the study earlier? As the \E variables we propose are not exchangeable, i.e. their values change under permutations of the data sequences, a direct comparison to the results of the SWEPIS study is not possible as the exact data stream is not available. To simulate a \quotec{real-time} scenario equivalent to the SWEPIS study, we assume we collect a total of $1380$ data blocks, with $n_a = n_b = 1$, with a total of $2760$ observations. We already know that in group a, $0$ events are observed. In group b, $6$ events are observed, of which we know that the last event was observed in data block $1380$, directly before the study was stopped. Hence, we can simulate the \quotec{real-time} data by permuting the indices of the observations in group b in the $1379$ first data blocks.

Four different approaches for analyzing the data with \E variables were explored: without any restriction on \hmc{1}, with a restriction based on the additive divergence measure (the minimal difference between the groups), with a restriction based on the log odds ratio, and with a restriction on the event rate in the control group \emph{and} on the minimal difference. The minimal difference, log odds ratio and event rate used were chosen based on a large recent meta-analysis on stillbirths \citep{muglu2019risks}; we used $\delta = 0.00318$ as a restriction on the difference between the groups, $\log(2)$ for the log odds ratio and $0.0001$ as the event rate. For all \E variables, the default beta prior hyperparameters with $\gamma = 0.18$ as earlier were used.

In Figure \ref{fig:swepis} the spread of the evidence collected with the four types of \E variables in $1000$ simulations analogous to the SWEPIS setting is depicted. Because the observed effect size was higher than expected, \E values obtained with the (too low) restriction on the effect size were lower than the \E values obtained with the \E variable without restrictions. Adding the restriction on the event rate increased the \E values, and in all 1000 simulations, the SWEPIS study would have been stopped before the occurrence of the sixth stillbirth. Figure \ref{fig:swepis} also depicts results of a second simulation experiment, where we sampled $1000$ data streams from $P_{0, 6/1380}$ and recorded the stopping times while analyzing the streams with the four \E variables with different restrictions on \hmc{1}. With the \E variables without restriction, or with a restriction on the event rate and difference between the groups, we would have often stopped data collection earlier than in the SWEPIS setting. 

We can thus conclude that, would the monitoring of the study have been performed with \E variables instead of p-values, first of all we would have collected \emph{correct} evidence for a higher proportion of stillbirths in the 42-weeks group, and second, the degree of evidence is quite similar to that collected with the (incorrectly determined) p-value: both are significant at the $0.05$ level. Wagemakers and Ly with their method also found evidence for the existence of a difference between the two groups, but not nearly of the same degree: they reported Bayes factors that varied, depending on the choice of the prior, between $1$ and $5.4$ (note that whenever we reject, our product of \E values, which like a Bayes factor can be thought of as a prequential likelihood ratio, must be $\geq 20$). A possible explanation for this difference could be that the Bayes factors used for collecting evidence in their study are not designed for analyzing stream data. As we also saw in our experiments, choosing the wrong prior or restriction on \hmc{1} can make a large difference for the evidence collected. These results show that when planning a prospective study, using \E variables for analysis could, through their flexibility, contribute to earlier evidence collection compared to existing methods.\\ 

\begin{figure}[H]
    \centering
    \begin{subfigure}{11cm}
         \centering
         \includegraphics[width=\textwidth]{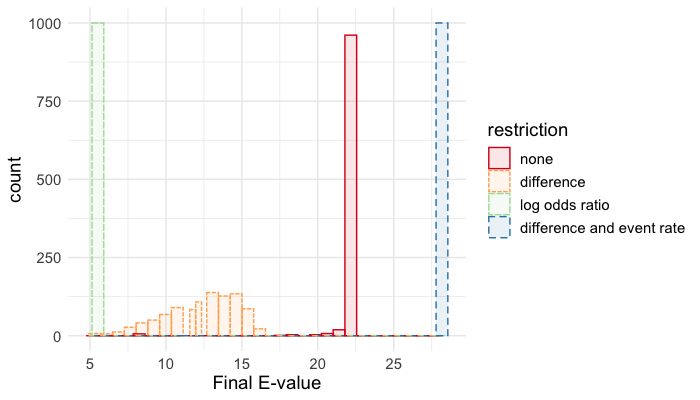}
         \caption{Simulated \E values in SWEPIS setting, stopping at $m = 1380$ or when $E \geq 20$}
    \end{subfigure}
    \hfill
    \begin{subfigure}{11cm}
     \centering
     \includegraphics[width=\textwidth]{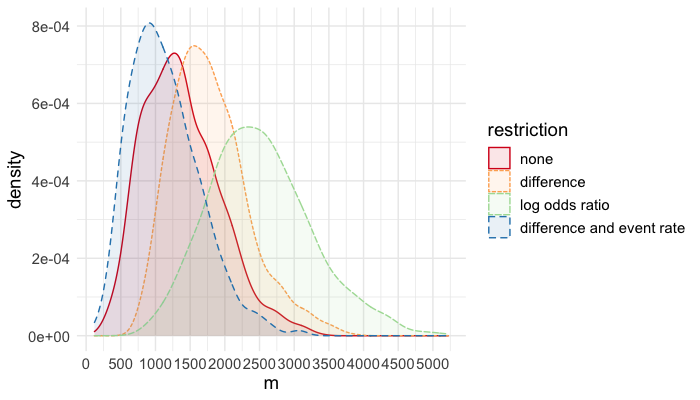}
     \caption{Simulated stopping times in setting with continuing until $E \geq 20$}
    \end{subfigure}
    \caption{Spread of \E values and stopping times observed with safe analysis of 1000 simulations of data streams analogous to the SWEPIS scenario, with four different types of restrictions on \hmc{1}.} 
    \label{fig:swepis}
\end{figure}

\section{Other \E Variables for Two Data Streams}
\label{sec:otherevariables}
\subsection{The GRO \E variable for some Exponential and Location Families} 
\label{sec:expfam}
The simplification (\ref{eq:simpleevarb}) shows that in the Bernoulli case with simple $\Theta_1 = \{(\theta^*_a,\theta^*_b)\}$, we can take in our denominator $p_{\theta_0}$ with  $\theta_0 = \frac{n_a}{n} \theta^*_a + \frac{n_b}{n} \theta^*_b$ --- which can also be interpreted as the distribution in the null corresponding to a mixture of the means, rather than the mixture of two distributions in the null. The Bernoulli model is a special case of  1-parameter exponential families which can all be parameterized in terms of their means so that $\Theta \subset \reals$ and ${\bf E}_{P_{\theta}}[Y] = \theta$; this is also possible for some location families that are not of exponential form. This suggests that, for all such models, instead of  (\ref{eq:simpleevar})  we might also consider the likelihood ratio (\ref{eq:simpleevarb}). For the Bernoulli model, both definitions will coincide, but for general 1-parameter exponential families they do not since their corresponding set of densities is not convex. The question is now whether (\ref{eq:simpleevarb}) defines an \E variable for general exponential families. It turns out that the answer is {\em no\/} in general, but {\em yes\/} in some special cases.
For a negative example, consider the case with $\Theta = \reals^+$ representing the family of exponential distributions in their mean-value parameterization, i.e. $p_{\theta}(y) = \lambda \exp(- \lambda y)$ with $\lambda =  1/\theta$ and take $n_a=n_b = 1$. A simple calculation shows that for any $\theta^*_a \neq \theta^*_b \in \Theta$, we have $\lim_{\theta \rightarrow \infty} {\bf E}_{Y_a, Y_b \text{\ i.i.d.}\sim P_{\theta}}[p_{\theta^*_a}(Y_a)p_{\theta^*_b}(Y_b)/p_{(\theta^*_a+\theta^*_b)/2}(Y_a,Y_b)]= \infty$. The negative binomial families provide, by a similar calculation, another negative example. For a positive example, consider the case with $\Theta = \reals$ representing the Gaussian location family with fixed variance $1$ and 
again take $n_a=n_b=1$. A simple calculation shows that (\ref{eq:simpleevarb}) is equal to the likelihood ratio for testing whether the difference $Z = Y_a-Y_b$ is a Gaussian with variance $\sqrt{2}$ with either mean $0$ or mean $\theta_b - \theta_a$. This is in fact the standard paired-sample $Z$-test that would normally be advised in this situation. In fact it is the GRO \E variable for this situation:
\begin{proposition}\label{prop:exp}
Let $\{P_{\theta}: \theta \in \Theta \}$ represent a family of probability distributions with densities $p_{\theta}$, with $\Theta$ a convex set in $\reals^k$ for some $k\geq 1$. For any $\theta^*_a,\theta^*_b \in \Theta$ we have: if (\ref{eq:simpleevarb}) is an \E variable for $\Theta_1 =\{ (\theta^*_a,\theta^*_b) \}$ then it is the GRO \E variable for $\Theta_1 =\{ (\theta^*_a,\theta^*_b) \}$.
\end{proposition}
The proof is immediate from Proposition~\ref{prop:proper}. 
The proposition implies that in the special cases in which (\ref{eq:simpleevarb}) does provide an \E variable, it is to be preferred (achieves better growth) above our original construction (\ref{eq:simpleevar}). (\ref{eq:simpleevar}) has the advantage that it provides an \E variable relative to arbitrary models. We plan to study the cases in which (\ref{eq:simpleevarb}) can be used instead in future work. 

\subsection{The Conditional \E variable for Tests of Two Proportions}
\label{sec:conditional}
\cite{Wald47} proposed a 2-sample sequential probability ratio test (SPRT) for the $2 \times 2$ setting. Since SPRTs can be written in terms of products of \E variables (although products of \E variables often do not give SPRTs; see the discussion by \cite{grunwald2019safe}), let us see what  \E variables Wald's test corresponds to.
The setting is restricted to size-2 blocks with $n_a = n_b = 1$. We measure effect size with $d$ the log-odds ratio (\ref{eq:logodds}) and consider an alternative with a $d(\theta_a,\theta_b)$ that is at least some given $\delta$. Using that, for all $(\theta_a,\theta_b) \in (0,1)^2$, $z \in \{0,1,2\}$, the conditional probability mass function $p_{\theta_a,\theta_b} (Y_a, Y_b \mid \sum Y_a +Y_b=z)$ only depends on the log-odds ratio, we can write it, as $q_{\delta}(y_a,y_b |z)$ where $q_{\delta}$ is a probability mass function whose definition depends on $(\theta_a,\theta_b)$ only via $\delta = d( (\theta_a,\theta_b))$.
We then take as our \E variable $S_{\textsc{cond},\delta}:= q_{\delta}(Y_a,Y_b \mid Y_a + Y_b)/q_0(Y_a,Y_b \mid Y_a + Y_b)$. Since the conditional distribution $q_0(Y_a, Y_b \mid Z)$ is the same for all distributions in the null, this conditional likelihood gives an \E variable and can be used instead of our generic \E variable. Since for this Bernoulli case, our \E variable is in fact GRO, we would expect this new conditional \E variable to perform worse in terms of GRO (and for the reasons given in Section~\ref{sec:basics} also in terms of the amount of data needed before one can reject at a desired power), and experiments (not reported here) 
confirm that it indeed performs slightly worse for $\delta$ close to $0$, and substantially worse for larger $\delta$. This is already suggested by the fact that, unlike the GRO \E variable,  $S_{\textsc{cond},\delta}$ takes on value $1$ whenever $y_a = y_b$, effectively ignoring data blocks in which both outcomes are the same. Another disadvantage is that it can only be used in combination with effect size given by the odds ratio or any monotonic transformation thereof; whereas the GRO \E variable can also be combined with the difference $\theta_b - \theta_a$ or any other desirable notion of effect size.  
\section{Conclusion}
We have established \E variables and  test martingales for the general two-i.i.d.-data streams problem. We have demonstrated, using theory, simulations and a real-world example that, for tests of two proportions, by choosing an appropriate prior on $\Theta_1$, the method can be made competitive with classical methods that do not allow for optional stopping. 
Whereas in this paper, we have focused on testing, our  \E variables can also be extended to get {\em anytime-valid confidence sequences} \citep{howard2018uniform,lai1976confidence}, i.e. confidence sequences for effect sizes that are valid even under optional stopping. This requires us to first extend the testing to scenarios with $\delta \geq \delta_1$ vs. $\delta \leq \delta_0$ for $\delta_0 \neq 0$, that is, null hypotheses with $\theta_a \neq \theta_b$. We will report on this extension elsewhere. Our work also suggests a question for future work that is practically relevant, easy to state but hard to answer:  to what extent do our findings generalize to logistic regression? 
\section*{Acknowledgements}
The authors gratefully acknowledge Reuben Adams, Rianne de Heide, Wouter Koolen, Muriel Perez, Judith ter Schure and Akshay Balsubramani for useful conversations and in particular Adams and De Heide for performing experiments that inspired the \E variables presented here. This work is part of the Enabling Personalized Interventions (EPI) project, which is supported by the Dutch Research Council (NWO) in the Commit2 - Data –Data2Person program under contract 628.011.028. Declarations of interest: none.

\section*{Supplementary material}
\begin{itemize}
\item Appendix \ref{app:numerical}: detailed description of numerical approach to calculating \E variables for restricted \hmc{1}\\
\item Appendix \ref{app:gd}: detailed description of Gunel-Dickey Bayes factors \\ 
\item Appendix \ref{app:experiments}: optional stopping experiments \\
\end{itemize}

\bibliography{references}
\section*{Appendix: Proofs}
The proofs below repeatedly use Theorem 1 of \cite{grunwald2019safe} and a direct corollary (called Corollary 2 by \cite{grunwald2019safe}), which we re-state here, for convenience, combined as a single statement.
We use the notation adopted later in the paper: for  $\cH_0 = \{P_{\theta}: \theta \in \Theta_0\}$ and, for $W$ a distribution on $\Theta_0$, we write $P_W = \int P_{\theta} d W(\theta)$.
\paragraph{Theorem (Theorem 1 of \cite{grunwald2019safe})} 
Let $Y$ be a random variable taking values in a set ${\cal Y}$. Suppose $Q$ is a probability distribution for $Y$ with density $q$ that is strictly positive on all of ${\cal Y}$ and let $\cH_0 = \{ P_{\theta}: \theta \in \Theta_0 \}$  be a set of distributions for $Y$ where each $P_{\theta}$ has density $p_{\theta}$. Let ${\cal W}_0$ be the set all distributions on $\Theta_0$. Assume  $\inf_{W_0 \in {\cal W}_0(\Theta_0)} D(Q  \| P_{W_0}) < \infty$.  Then (a) 
there exists a (potentially sub-) distribution $P^*_0$ with density $p^*_0$ such that 
\begin{equation*}
S^{*} := \frac{q(Y)}{p^*_0(Y)}
\end{equation*} 
is an \E variable ($p^*_0$ is called the {\em Reverse Information Projection (RIPr) of $q$ onto $\{ p_{W}: W \in {\cal W}_0 \}$}
\citep{XLi99,XLiB00,Xgrunwald2019safe}).
Moreover, (b), $S^*$ satisfies 
\begin{align}\label{eq:firstgro}
  \sup_{S \in {\cal E}(\Theta_0)}  {\bf E}_{Y \sim Q}[\log S]
  =  
  {\bf E}_{Y \sim Q}[\log S^*] =
  \inf_{W_0 \in {\cal W}_0(\Theta_0)}  D(Q \| P_{W_0})
  = D(Q \| P^*_0).  
  \end{align} 
  and is thus the $Q$-GRO \E variable for $Y$. 
If the minimum is achieved by some ${W}^*_0$, i.e.\ $D(Q \| P^*_0) = D(Q \|
  P_{W^*_0})$, then $P^*_0 = P_{W^*_0}$. 
Moreover, (c), if there exists an \E variable $S$ of the form $q(Y)/p_{W_0}(Y)$ for some $W_0 \in {\cal W}_0$ then $W_0$ must achieve the infimum in (\ref{eq:firstgro}) and $S$ must be essentially equal to $S^*$ in the sense that for  all $P \in \cH_0 \cup \{Q \}$, $P(S^* = q(Y)/p_{W_0}(Y)) = 1$. Similarly (d), if there exists a $W^*_0 \in {\cal W}_0$ that achieves the infimum in (\ref{eq:firstgro}) then $S = q(Y)/p_{W^*_0}(Y)$ is an \E variable and $S$ is again essentially equal to $S^*$.  
\subsection{Proof of Propositions}
\paragraph{Proof of Proposition \ref{prop:proper}} 
Below we state and prove a slight generalization of Proposition~\ref{prop:proper}. 
\begin{proposition}\label{prop:properb}
Let $\cH_1 = \{Q\}$ be a singleton and let $\cH_0 = \{P_{\theta}: \theta \in \Theta_0 \}$ be such that for some distribution $W$ on $\Theta_0$, $D(Q \| P_{W}) < \infty$.  
For general $\theta \in \Theta_0$ and distributions $W$ on $\Theta_0$, define $S_{\theta, (j)} := q(Y_{(j)})/p_{\theta}(Y_{(j)})$ and $S_{W,(j)} = 
q(Y_{(j)})/p_{W}(Y_{(j)}).$
We have: 
\begin{enumerate}
    \item Suppose there exists a distribution $W$ on $\Theta_0$ such that $S_{W, (1)}$ is an \E variable. Then  
$S_{W,(1)}$ is the $Q$-GRO \E variable for $Y_{(1)}$. In particular, if $W$ puts mass $1$ on a particular $\theta^{\circ} \in \Theta_0$, then $S_{W,(1)} = S_{\theta^{\circ},(1)}$ is the $Q$-GRO \E variable. 
\item If $\Theta_0 = \{\theta_0\}$ is simple then, with the prior $W_0$ putting mass 1 on $\theta_0$, $S_{W_0,(1)}= S_{\theta_0,(1)}$ is an \E variable and hence, by the above, also the $Q$-GRO \E variable.  
\item If, for some $\theta^{\circ} \in \Theta_0$, $S_{\theta^{\circ}, (1)}$ is an \E variable and we further assume that $Y_{(1)}, Y_{(2)}, \ldots$ are i.i.d. according to all distributions in $\cH_0 \cup \cH_1$, then $S^{(m)}_{\gro(Q)} = \prod_{j=1}^m S_{\theta^{\circ}, (j)}$; that is, the $Q$-GRO optimal (unconditional) \E variable for $Y^{(m)}$ is the product of the individual $Q$-GRO optimal \E variables. 
\end{enumerate}
\end{proposition}
\begin{proof}
{\em Part 1\/} The theorem above, part (b), implies, with $Y = Y_{(1)}$, that some  $Q$-GRO \E variable $S^*$ for $Y_{(1)}$ exists. Part (c) then implies that 
we can take $S^*$ to be equal to $S_{W, (1)}$. This implies the statement. 

{\em Part 2\/} is immediate. 
{\em Part 3\/}
We assume that  $S_{\theta^{\circ},(1)}$ is an \E-variable. Then  the i.i.d. assumption implies that  $S_{\theta^{\circ}}^{(m)} := \prod_{j=1}^m S_{\theta^{\circ},(j)} = \prod q(Y_{(j)})/p_{\theta^{\circ}}(Y_{(j)})$ is also an \E variable. But \citep[Theorem 1]{Xgrunwald2019safe}, part (c) as stated above implies (by taking a distribution $W$ putting mass $1$ on $\theta$) that for $\cH_0$ for which data are i.i.d., for each $m\geq 1$, that  if a $\theta \in \Theta_0$ exists such that $S_{\theta}^{(m)}$ is an \E variable, then $S_{\theta}^{(m)}$ must be the $Q$-GRO \E variable for $Y^{(m)}$. This proves the  statement.
\end{proof}
\paragraph{Proof of Proposition \ref{prop:beta}}
The formulae for $\breve\theta_a |Y^{(j-1)}$ and $\breve\theta_b |Y^{(j-1)}$
are standard expressions for the Bayes predictive distribution based on the given beta priors; we omit further details. As to the expression for $\breve\theta_0|Y^{(j-1)}$ in terms of $\kappa = n_b / n_a$: 
Straightforward rewriting gives, for general $\alpha_a,\alpha_b, \beta_a,\beta_b$:
\begin{equation}\breve\theta_0 |Y^{(j-1)} = \frac{1}{1 + \kappa}\breve\theta_a |Y^{(j-1)} 
+ \frac{\kappa}{1 + \kappa} \breve\theta_b |Y^{(j-1)}.\end{equation} 
If we plug in the expressions for $\breve\theta_a |Y^{(j-1)}, \breve\theta_b |Y^{(j-1)}$ and we instantiate to  
$
\alpha_b = \kappa \alpha_a \text{, and } \beta_b = \kappa \beta_a,
$
this becomes
\begin{align*}
    \breve\theta_0 |Y^{(j-1)}  &= \frac{1}{1 + \kappa}\frac{U_a + \alpha_a}{n_a(j-1) + \alpha_a + \beta_a} + \frac{\kappa}{1 + \kappa} \frac{U_b + \alpha_b}{\kappa(n_a (j-1) + \alpha_a + \beta_a)} \\
    &= \frac{1}{1 + \kappa}\frac{U_a + U_b + (1 + \kappa) \alpha_a}{n_a (j-1) + \alpha_a + \beta_a} 
    = \frac{U + (1 + \kappa) \alpha_a}{n (j-1) + (1+\kappa)\alpha_a + (1+\kappa)\beta_a},
\end{align*}
which is what we had to prove. 
\subsection{Proof of Theorem \ref{simpleSproportions}}
We first restate  Theorem \ref{simpleSproportions} in its extended version that holds for $k\geq 2$ data streams. Let $\vec{n}= (n_1,\ldots, n_k), n = \sum_{g=1}^k n_g, \vec{\theta}= (\theta_a,\ldots, \theta_k) \in \Theta^k$ and $\vec{y}^n$ be as above (\ref{eq:evarkextension}). We use `$\vec{Y}^n \sim P_{\theta^*}$'  as an abbreviation for ` $Y_1^{n_1} \sim P_{\theta^*_1};
\ldots ; Y_k^{n_k} \sim P_{\theta^*_k}
$'.
\setcounter{theorem}{0}
\begin{theorem}
\label{simpleSproportionsAppendix}
Let 
\begin{equation*}
s(\vec{y}^n; \vec{n}, \vec{\theta}^*) 
\coloneqq \prod_{g = 1}^k
\frac{p_{\theta_{g}^*}( y^{n_g}_g) }{
\prod_{i=1}^{n_g} 
 \left( \sum_{g' = 1}^k \frac{n_{g'}}{n} p_{\theta^*_{g'}} (y_{i,g}) \right)} .
 \end{equation*}
The  random variable $S_{[\vec{n}, \vec{\theta}^*]} :=  s(\vec{Y}^n; \vec{n}, \vec{\theta}^*)$ is an \E variable, i.e. 
we have: 
\begin{equation*}
\sup_{\theta \in \Theta} {\bf E}_{V^n \sim  P_{\theta}}\left[
s( V^n ; \vec{n}, \vec{\theta}^*)
\right] \leq 1.
\end{equation*}
Moreover, if $\{ P_{\theta}: \theta \in \Theta \}$  is a  convex set of distributions, then $S_{[\vec{n}, \vec{\theta}^*]}$ is the $(\vec{\theta}^*)$-GRO \E variable: for any non-negative function $s'$ on $\cY^{n}$ satisfying
$\sup_{\theta \in \Theta} {\bf E}_{V^n \sim  P_{\theta}}\left[
s'( V^n)\right] \leq 1$, we have:
\begin{equation*}
{\bf E}_{\vec{Y}^n \sim P_{\theta^*}}[ \log 
s(\vec{Y}^n; \vec{n}, \vec{\theta}^*) ] \geq 
{\bf E}_{\vec{Y}^n \sim P_{\theta^*}}[ \log s'(\vec{Y}^n)].
\end{equation*}
\end{theorem}
\paragraph{Proof of Theorem \ref{simpleSproportionsAppendix}} 
The following fact plays a central role in the proof:
\paragraph{Fact} For $g \in (1, ..., k)$, let $n_g \in \naturals, n := \sum_{g = 1}^k n_g$ and let $u_g \in \reals^+$. Suppose that $\sum_{g = 1}^k n_g u_g \leq n$.
Then $\prod_{g = 1}^k u_g^{n_g} \leq 1$. \\ \ \\ \noindent
This result follows  from the following standard generalization of Young's inequality to $k$ numbers: for any $k$ numbers $u_1,\ldots, u_k \in \reals^+_0$ and any $k$ nonnegative numbers $p_1,\ldots, p_k$ with $\sum_{g=1}^k p_g = 1$, we have  $\prod_{g=1}^k u_g^{p_g} \leq \sum_{g=1}^k p_g u_g$. Applying this with $p_g=n_g/n$ to $u_g$ and $n_g$ as above, we get   $\prod_{g=1}^k u_g^{n_g/n} \leq \sum_{g=1}^k (n_g u_g)/n \leq 1$, and the result follows by exponentiating to the power $n$.\\
{\em Part 1\/}
For $y \in \cY$, set set $p^{\circ}(y):= \sum_{g = 1}^k (n_g/n) p_{\theta^*_g}(y)$ and $p^{\circ}(y^m) = \prod_{i=1}^m p^{\circ}(y_i)$. For all $\theta \in \Theta$ we have:
\begin{align}\label{eq:kvoorschieten}
  &  {\bf E}_{V^n \sim P_{\theta}}\left[
  s( V^n ; \vec{n}, \vec{\theta}^*)
  \right] = 
\prod_{g = 1}^k {\bf E}_{
  Y_g^{n_g} \sim P_\theta
  }\left[
  \frac{p_{\theta_{g}^*}( Y^{n_g}_g) }{
p^{\circ}(Y^{n_g}_g)}
  \right] = \prod_{g = 1}^k \left({\bf E}_{Y \sim  P_{\theta}}\left[
  \frac{p_{\theta_{g}^*}( Y) }{
p^{\circ}(Y)}
  \right] \right)^{n_g}.
\end{align}
We also have
\begin{align}\label{eq:kinkoppen}
& \sum_{g = 1}^k \frac{n_g}{n} {\bf E}_{Y \sim  P_{\theta}}\left[
  \frac{p_{\theta_{g}^*}( Y) }{
p^{\circ}(Y)}
  \right] =  {\bf E}_{Y \sim  P_{\theta}}\left[ \sum_{g = 1}^k \frac{n_g}{n} \cdot 
  \frac{p_{\theta_{g}^*}( Y) }{ \sum_{g' = 1}^k
\frac{n_{g'}}{n} p_{\theta^*_{g'}}(Y)
}  \right] 
  = 1.
\end{align}
The result now follows by combining (\ref{eq:kvoorschieten}) with (\ref{eq:kinkoppen}) using the Fact further above.  \\
{\em Part 2 \/}
By convexity of $\{P_{\theta}: \theta \in \Theta \}$, there exists $\theta^{\circ} \in \Theta$ such that   $p_{\theta^{\circ}} = \sum_{g = 1}^k (n_g/n) p_{\theta^*_g}$ 
and then the numerator in (2.2) can we rewritten as $p_{\theta^{\circ}} (\vec{y})$. 
The GRO-property is now an immediate consequence of Proposition~\ref{prop:properb}, Part 1.

\newpage
\begin{appendices}
\appendixpage
\renewcommand{\thesection}{S\arabic{section}}
\renewcommand\thefigure{\thesection.\arabic{figure}}  
\section{Numerical approach to calculating \E variables for restricted \hmc{1}}\label{app:numerical}
In this subsection we describe how we propose to approximate the beta prior and posterior on the restricted \hmc{1} with parameter space $\Theta(\delta)$, as defined in (4.1). Note that we limit ourselves to $\delta > 0$ in this detailed description; for $\delta < 0$ one can apply an entirely equivalent approach, with an extra term in the reparameterization. We define
\begin{equation*}
    \zeta = 
    \begin{cases}
    \delta \text{ if } d((\theta_a, \theta_b)) = \theta_b - \theta_a, \\
    0 \text{ if } d((\theta_a, \theta_b)) = \text{log-odds-ratio}(\theta_a, \theta_b),\\
    \end{cases}
\end{equation*}
such that we have $\theta_a \in (0, 1 - \zeta)$ and in both cases, $\theta_b$ is completely determined by $\theta_a$: $\theta_b = d^{-1}(\delta; \theta_a)$. Hence, our density estimation problem now becomes one-dimensional, which enables us to put a discretized prior on the restricted parameter space. 

First, we discretize the parameter space $\Theta_a$ to a grid (a vector) with precision $K, K \in (0, 1 - \zeta)$ and $1/K \in \mathbb{N}^+$: $\bm{\bar{\theta}_a} = \left(K, 2K, 3K, \ldots, 1 - \zeta\right)$. Then, we reparameterize $\theta_a = (1 - \zeta)\rho$, with $\rho \in (0,1)$. Then, we have 
$\bm{\bar{\rho}} = \left(K/(1 - \zeta), 2K/(1 - \zeta), \ldots, 1\right).$
For the discretized grid $\bm{\bar{\rho}}$, we compute the prior $W = \text{Beta}(\alpha, \beta)$ densities and normalize them, which also gives us the discretized densities for each $\theta_a^i \in \bm{\bar{\theta}_a}$ (with $i \in (1, 2, \dots, 1/K)$):
$$\pi_{\alpha, \beta, \zeta}(\theta_a^i) = \frac{\text{Beta}(\frac{\theta_a^i}{1 - \zeta}; \alpha, \beta)}{\sum_{k = 1}^{\frac{1}{K}}\text{Beta}(\frac{\theta_a^k}{1 - \zeta}; \alpha, \beta)}.$$
For all elements of $\bm{\bar{\theta}_a}$, the corresponding $\theta_b$ is retrieved and the likelihood of incoming data points $p_{\theta_a, \theta_b}(Y^{(j-1)})$ is calculated. We can then estimate the posterior density of $\theta_a^i \in \bm{\bar{\theta}_a}$:
$$p(\theta_a^i | Y^{(j-1)}) = \frac{\pi_{\alpha, \beta, \zeta}(\theta_a^i) p_{\theta_a^i, \theta_b^i}(Y^{(j-1)})}{\sum_{k = 1}^{\frac{1}{K}} \pi_{\alpha, \beta, \zeta}(\theta_a^k)p_{\theta_a^k, \theta_b^k}(Y^{(j-1)})}.$$
We can then estimate $\breve\theta_a | Y^{(j-1)} = {\bf E}_{\theta_a \sim W \mid Y^{(j-1)}}[\theta_a]$ as
$\sum_{i = 1}^{\frac{1}{K}} p(\theta_a^i | Y^{(j-1)}) \theta_a^i,$ and $\breve\theta_b | Y^{(j-1)} = d^{-1}(\delta; \theta_a | Y^{(j-1)})$.
\newpage
\section{The Gunel-Dickey Bayes Factors do not give rise to \E-variables}
\label{app:gd}
\setcounter{figure}{0}   
\begin{table}[ht]
    \centering
    \begin{tabular}{l|l|l}
         Sampling scheme & Fixed parameters & Bayes factor (10) for 2x2 table\\
         \hline
        Poisson & none & $\frac{8(n+1)(n_1 + 1)}{(n+4)(n+2)} \left[\frac{n_{a1}! n_{b1}! n_{a0}! n_{b0}! n!}{(n_1 + 1)! n_0! n_a! n_b!}\right]$\\
        Joint multinomial & n & $\frac{6(n+1)(n_1 + 1)}{(n+3)(n+2)} \left[\frac{n_{a1}! n_{b1}! n_{a0}! n_{b0}! n!}{(n_1 + 1)! n_0! n_a! n_b!}\right]$\\
        Independent multinomial & $n_a$, $n_b$ & $\frac{{n \choose  n_1}}{{n_a \choose n_{a1}} {n_b \choose n_{b1}}} \frac{(n + 1)}{(n_a + 1) (n_b + 1)}$\\
        Hypergeometric & $n_a$, $n_b$, $n_1$ & $\frac{n_{a1}! n_{b1}! n_{a0}! n_{b0}! n!}{\prod_{i \in \{a,b, 0, 1\}} (n_i + \mathbb{I} {n_i = min(n_a, n_b, n_0, n_1)})!}$\\
    \end{tabular}
    \caption{Overview of (objective) Bayes factors for contingency table testing provided by \cite{XDickey1974} and \cite{Xjamil2017default}.
    \label{tab:my_label}}
\end{table}
We will not consider the hypergeometric and joint multinomial scenarios for this paper, where the number of successes $n_1$ is fixed, as they do not match the block-wise data design in this paper. The Bayes factor for the Poisson sampling scheme is not an \E variable, as the expectation under the null hypothesis with Poisson distributions on individual cell counts exceeds $1$ for rates $\lambda \geq 1$:

\begin{align*}
    & \mathbb{E}_{n_{rc} \sim \text{Poisson}(\lambda_{rc})} \left[ BF_{10}(N_{a1}, N_{b1}, N_{a0}, N_{b0})\right] = \\
    & \sum_{n_{a1} = 0}^{\infty} \ldots \sum_{n_{b0} = 0}^{\infty} \pi_{\lambda_{a1}}(n_{a1}) \ldots \pi_{\lambda_{b0}}(n_{b0})BF_{10}(n_{a1}, n_{b1}, n_{a0}, n_{b0}) = \\
    & \frac{8}{\exp({\lambda_{a1} + \ldots + \lambda_{b0}})} 
    \sum_{n_{a1} = 0}^{\infty} \ldots \sum_{n_{b0} = 0}^{\infty}
    \lambda_{a1}^{n_{a1}} \ldots \lambda_{b0}^{n_{b0}}\frac{(n+1)(n_1 + 1)}{(n+4)(n+2)} \frac{n!}{(n_1 + 1)! n_0! n_a! n_b!},
\end{align*}
as illustrated numerically in Figure \ref{fig:poissonGD} for increasing limits for the sums $\sum_{n_{rc} = 1}^{\max n_{rc}}$.
\begin{figure}[H]
    \centering
    \begin{subfigure}{6cm}
      \centering
      \includegraphics[width=\textwidth]{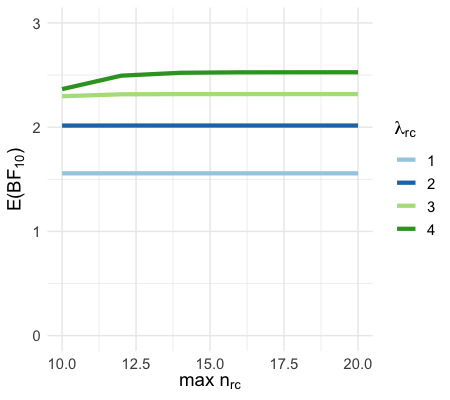}
      \caption{The Gunel-Dickey Bayes factor for the Poisson sampling scheme is not an \E variable: $\sum_{n_{a1} = 0}^{\max n_{rc}} \ldots \sum_{n_{b0} = 0}^{\max n_{rc}} \pi_{\lambda_{a1}}(n_{a1}) \ldots \pi_{\lambda_{b0}}(n_{b0})BF_{10}(n_{a1}, n_{b1}, n_{a0}, n_{b0})$ for various $\max n_{rc}$ and $\lambda_{rc}$.}
    \end{subfigure}
    \hfill
    \begin{subfigure}{6cm}
      \centering
      \includegraphics[width=\textwidth]{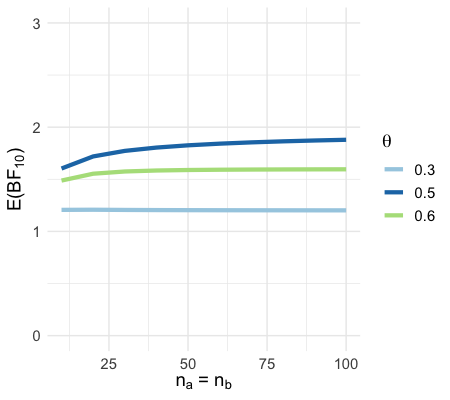}
          \caption{The Gunel-Dickey Bayes factor for the independent multinominal sampling scheme is not an \E variable: $\mathbb{E}_{N_{a1}, N_{b1} \sim \text{Binomial}(\theta)} \left[ BF_{10}(N_{a1}, N_{b1} | n_a, n_b) \right]$ for various choices of $\theta$ and $n_g$.}
    \end{subfigure}
    \caption{GD }
    \label{fig:poissonGD}
\end{figure}

For the independent multinomial sampling scheme, let, without loss of generality, $n_a < n_b$. We get, with $n_0 = n - n_1$,
\begin{align*}
    & \mathbb{E}_{N_{a1}, N_{b1} \sim \text{Binomial}(\theta)} \left[ BF_{10}(N_{a1}, N_{b1} | n_a, n_b) \right] = \\
    & \sum_{n_{a1} = 0}^{n_a} \sum_{n_{b1} = 0}^{n_b} 
    {n_a \choose n_{a1}}  
    {n_b \choose n_{b1}} \theta^{n_1} (1 - \theta)^{n_0}
    \frac{{n \choose  n_1}}{{n_a \choose n_{a1}} {n_b \choose n_{b1}}} \frac{(n + 1)}{(n_a + 1) (n_b + 1)} = \\
    & \frac{(n + 1)}{(n_a + 1) (n_b + 1)} \sum_{n_{a1} = 0}^{n_a} \sum_{n_{b1} = 0}^{n_b} 
    {n \choose  n_1} 
    \theta^{n_1} (1 - \theta)^{n_0} 
\end{align*}
Numerical simulations show that, for a range of choices for $n, n_a$ and $\theta$ this exceeds 1; see Figure \ref{fig:poissonGD}.

\newpage
\section{Type-I error guarantee under optional stopping}
\setcounter{figure}{0}   
\label{app:experiments}
\paragraph{Type-I Error}\label{sec:typeIexperiment}
In Figure \ref{fig:type1Error} type-I error rates of several \E variables and Fisher's exact test estimated through a simulation experiment are depicted. $2000$ samples of length $1000$ were drawn according to a Bernoulli$(0.1)$ distribution to represent $1000$ data streams in two groups. After each complete block $m \in \{1, \dots, 1000\}$ an \E value or p-value was calculated and the proportion of rejected experiments up until $m$ with each test type was recorded. As the stream lengths increase, the type-I error rate under (incorrectly applied) optional stopping with Fisher's exact test increases quickly. The type-I error rate of the \E variables remains bounded.

\begin{figure}[ht]
    \centering
    \includegraphics[width = 11cm]{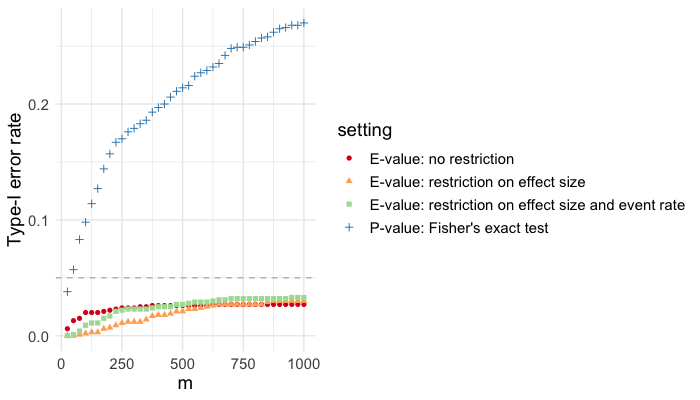}
    \caption{Type-I error rates for various \E variables and Fisher's exact test under optional stopping estimated with $1000$ simulations of two Bernoulli$(0.1)$ data streams of length $1000$, with $n_a = n_b = 1$. Significance level $\alpha = 0.05$ was used (grey dashed line). For the safe tests, beta prior parameter values used were $\gamma= \alpha_a = \beta_a = \alpha_b = \beta_b = 1/2$ ($\gamma = 0.18$ gave comparable results). For the \E variables with restrictions on \hmc{1}, we used $\delta = 0.05$ and $\theta_a = 0.1$.}
    \label{fig:type1Error}
\end{figure}

\end{appendices}

\end{document}